\documentclass[11pt,draftclsnofoot,onecolumn,journal,letterpaper]{IEEEtran}
\usepackage{graphicx}
\usepackage{amsmath,amssymb,amsthm,mathrsfs,amsfonts,dsfont}
\usepackage[ruled]{algorithm2e}
\usepackage{url}
\usepackage{color}
\usepackage{subfigure}

\newtheorem{theorem}{Theorem}
\newtheorem{prop}{Proposition}

\newtheorem{defi}{Definition}

\ifodd 0
\newcommand{\com}[1]{{\color{red}#1}}
\else
\newcommand{\com}[1]{}
\fi

\ifodd 0

\else

\fi

\newcommand{\bs}{\boldsymbol}

\newcommand{\bb}{\mathbb}

\begin{document}

\title{A Non-stochastic Learning Approach to Energy Efficient Mobility Management}

\author{Cong~Shen,~\IEEEmembership{Senior Member,~IEEE,}
        Cem~Tekin,~\IEEEmembership{Member,~IEEE,}
        and~Mihaela~van~der~Schaar,~\IEEEmembership{Fellow,~IEEE}
\thanks{C. Shen is with the Department of Electronic Engineering and Information Science, University of Science and Technology of China. E-mail:  \texttt{congshen@ustc.edu.cn}.}
\thanks{C. Tekin is with the Department of Electrical and Electronics Engineering, Bilkent University, Ankara, Turkey, 06800. E-mail:  \texttt{cemtekin@ee.bilkent.edu.tr}.}
\thanks{M. van der Schaar is with the Electrical Engineering Department, University of California, Los Angeles (UCLA), USA. E-mail:  \texttt{mihaela@ee.ucla.edu}.}
}

\maketitle
\begin{abstract}

Energy efficient mobility management is an important problem in modern wireless networks with heterogeneous cell sizes and increased nodes densities.  We show that optimization-based mobility protocols cannot achieve long-term optimal energy consumption, particularly for ultra-dense networks (UDN). To address the complex dynamics of UDN, we propose a \textit{non-stochastic} online-learning approach which does not make any assumption on the statistical behavior of the small base station (SBS) activities. In addition, we introduce  \textit{handover cost} to the overall energy consumption, which forces the resulting solution to explicitly minimize frequent handovers.  The proposed Batched Randomization with Exponential Weighting (BREW) algorithm relies on \textit{batching} to explore in bulk, and hence reduces unnecessary handovers. We prove that the regret of BREW is sublinear in time, thus guaranteeing its convergence to the optimal SBS selection. We further study the robustness of the BREW algorithm to delayed or missing feedback. Moreover, we study the setting where SBSs can be dynamically turned on and off. We prove that sublinear regret is impossible with respect to arbitrary SBS on/off, and then develop a novel learning strategy, called {ranking expert} (RE), that simultaneously takes into account the handover cost and the availability of SBS. To address the high complexity of RE, we propose a {contextual ranking expert} (CRE) algorithm that only assigns experts in a given context. Rigorous regret bounds are proved for both RE and CRE with respect to the best expert. Simulations show that not only do the proposed mobility algorithms  greatly reduce the system energy consumption, but they are also robust to various dynamics which are common in practical ultra-dense wireless networks.
\end{abstract}


\section{Introduction}
\label{sec:intro}

The ultra-dense deployment of small base stations (SBS) \cite{Quek:13} introduces new challenges to the wireless network design. Among these challenges, mobility management has become one of the key bottlenecks to the overall system performance. Traditionally, mobility management was designed for large cell sizes and infrequent handovers, which works well with the RF-planned macro cellular networks. The industry protocol is simple to implement and offers reliable handover performance \cite{SesiaLTE}. However, introducing SBSs into the network drastically complicates the problem due to the irregular cell sizes, unplanned deployment, and unbalanced load distributions \cite{Andrews:14}. Furthermore, ultra-dense deployment makes the problem even harder, as user equipments (UE) in ultra-dense networks (UDN) can have many possible serving cells, and mobile UEs may trigger very frequent handovers even without much physical movement.  Simply applying existing macro solutions leads to a poor SBS mobility performance. In particular, total energy consumption can be significant when the mobility management mechanism is not well designed \cite{Xenakis:13}. 

To address these challenges, research on mobility management has recently attracted a lot of attention from both academia and industry \cite{Andrews:14}. The research has mainly been  based on \textit{optimization theory}, i.e., given the various UE and BS information, the design aims at maximizing certain system utility by finding the best UE-BS pairing. The problem is generally non-convex and optimal or suboptimal solutions have been proposed. In \cite{Ye:13}, a utility maximization problem is formulated for the optimal user association which accounts for both user's RF conditions and the load situation at the BS. For energy efficient user association, \cite{Mesodiakaki:14} aims at maximizing the ratio between the total data rate of all users and the total energy consumption for downlink heterogeneous networks. Althunibat et. al. \cite{Althunibat:14} propose a handover policy in which low energy efficiency from the serving BS triggers a handover, and the design objective is to maximize the achievable energy efficiency under proportionally fair access. Another optimization criterion of minimizing the total power consumption, while satisfying the user's traffic demand for UDN is considered in \cite{Bottai:14}. 

These existing solutions have been proved effective for less-densified heterogeneous networks, but they may perform poorly when the network density becomes high. Examples include the so-called frequent handover (FHO), Ping-Pong (PP), and other handover failures (HOF) problems, which commonly occur when the UE is surrounded by many candidate BSs \cite{3gpp.36.839}. In this scenario, a UE may select its serving SBS based on some optimization criterion, e.g., best biased signal strength as in 3GPP, or other system metric as in \cite{Ye:13,Mesodiakaki:14,Althunibat:14,Bottai:14}. However, system dynamics such as very small movement of the UE or its surrounding objects can quickly render the solution sub-optimal, triggering user handover procedure in a frequent manner. Probably more important than the loss of throughput, the FHO and PP problems significantly increase the system energy consumption, as much energy is wasted on unnecessary handovers.

These new problems have motivated us to adopt an \textit{online learning} approach, rather than an optimization one, to mobility management.  The rationale is that the goal of mobility should not be to maximize the \emph{immediate} performance at the time of handover, as most of the existing handover protocols do. Rather, mobility should build a UE-BS association that maximizes the \emph{long-term} performance. In fact, one can argue that the optimization approach with immediate performance maximization inevitably results in some of the UDN mobility problems such as FHO and PP. This is because the optimization-based solutions depend on the system information, and once the system configuration or the context information evolves, either the previously optimal solution no longer offers optimal performance\footnote{The industrial intuition is that good performance at the time of handover should carry over for the near future until the next handover is triggered. However, this is no longer valid with UDN, in which RF and load conditions can change dynamically and FHO/PP needs to be avoided.}, or a new optimization needs to run which can lead to increased energy consumption when the optimality criterion is frequently broken. Furthermore, optimization-based solutions rely on the accurate knowledge of various system information, which may not be available \textit{a priori} but must be learned over time.

The key challenge for efficient mobility management is the unavailability of accurate information of the candidate SBSs in an \emph{uncertain} environment. Had the UE known \textit{a priori} which SBS offers the best long-term performance, it would have chosen this SBS from the beginning and stuck to it throughout, thus avoiding the frequent handovers which lead to energy inefficiency while achieving optimal energy consumption for service. Without this omniscient knowledge, however, the UE has to balance immediate gains (choosing the current best BS) and long-term performance (evaluating other candidate BSs). Multi-armed bandit (MAB) can be applied to address such exploration and exploitation tradeoff that arises in the mobility learning problem, and there are a few works applying the \textit{stochastic} bandit algorithms \cite{Simsek:15,Capdevielle:12} to address this challenge. Mobility management in a heterogeneous network with high-velocity users is considered in \cite{Simsek:15}, where the solution uses stochastic MAB theory to learn the optimal cell range expansion parameter. In \cite{Capdevielle:12}, the authors propose a stochastic MAB-based interference management algorithm, which improves the handover performance by reducing the effective interference. 

There are three major issues in applying a stochastic bandit approach to the considered mobility management problem. Firstly, one must be able to assume that there exists a well-behaved stochastic process that guides the generation of the reward sequence for each SBS. In practice, however, it is difficult to unravel such statistical model for the reward distributions at SBS. Practical wireless networks with a moderate number of nodes or users are already complex enough that simple stochastic models, as  often used in stochastic bandit algorithms, cannot accurately characterize their behavior. Another problem is that the time duration within which a particular statistical model may be adequate is short due to high UDN system dynamics. As a result, there may not exist enough time to learn which statistical model to adopt, let alone utilize it to achieve optimal performance. Furthermore, in a practical system, there may be multiple UEs being served by one SBS, and the energy consumption depends not only on the SBS activity but also on the activities of other UEs, including their time-varying mobility decisions, traffic load, service requirement, etc. As a result, an \textit{uncontrolled} stochastic process cannot adequately capture  practical interactions between the UEs and the SBSs, and probabilistic modelling may not accurately match the real-world UDN energy behavior. Second, the majority of the stochastic MAB literature considers reward sequences that are generated by either an {\em independent and identically distributed} (i.i.d.) process \cite{lai1,auer}, or a  {\em Markov} process \cite{tekin2012online}. These restrictions may not accurately capture the SBS/UE behavior, resulting in a mismatch to the real-world performance. Lastly, the existing solutions cannot solve the FHO problem, because they do not consider the additional loss incurred when the UE performs handover from one SBS to another. In fact, most of the stochastic bandit solutions incur fairly frequent ``exploration'' operations, which directly lead to the FHO problem. 

Due to the aforementioned challenges that cannot be easily addressed by a stochastic approach, we opt out from using the stochastic MAB formulation. In this work, we solve the UDN mobility problem with the objective of minimizing long-term energy consumption, by using a \textit{non-stochastic} model. Specifically, we do not make any assumption on the statistical behavior of the SBS activities. Instead, the energy consumption of any SBS is allowed to vary arbitrarily. This is a fundamental deviation from the previous \textit{stochastic} solutions. A comparison of this work to the existing literature is provided in Table~\ref{tab:comp}. Note that the non-stochastic MAB problem is significantly harder than the stochastic counterpart due to the adversarial nature of the loss sequence generation. We  develop a new set of algorithms for loss sequences with switching penalties, delayed or missing feedback, and dynamic on/off behavior, and proved their effectiveness with both rigorous regret analysis and comprehensive numerical simulations. 

\begin{table*}
\caption{Comparison of our work with existing solutions.}
\label{tab:comp}
\centering
\begin{tabular}{|l|c|c|c|} 
\hline
{} &\cite{Ye:13,Althunibat:14,Mesodiakaki:14,Bottai:14} & \cite{Simsek:15,Capdevielle:12} & This work \\ \hline
\emph{Design tool} & Optimization & Stochastic learning & Non-stochastic learning \\ \hline
\emph{Optimize for energy consumption} & No & Yes & Yes \\ \hline
\emph{Distributed solution} & No, except \cite{Ye:13} & Yes & Yes \\ \hline
\emph{Solve FHO} & No & No & Yes \\ \hline
\emph{Forward-looking} & No & Yes & Yes \\ \hline
\emph{Robustness} & Not Considered & Not  Considered  & Considered \\ \hline
\emph {Performance study} & Analysis \& Simulation & Simulation & Analysis \& Simulation \\ \hline
\end{tabular}
\end{table*}

The main contributions of this paper are summarized below. 
\begin{itemize}
\item We propose a \textit{non-stochastic} energy consumption framework for mobility management. To the best of the authors' knowledge, this is the first work that applies the non-stochastic bandit theory to address mobility management in wireless networks. 
\item We explicitly add the \textit{handover cost} to the utility function to model the additional energy consumptions due to handovers, and thus force the optimal solution to minimize frequent handovers.
\item We present a Batched Randomization with Exponential Weighting (BREW) algorithm that addresses the frequent handover problem and achieves low system energy consumption. The performance of BREW is rigorously analyzed and a finite-time upper bound for performance loss due to learning is proved. We further study the effect of delayed or missing feedback and analyze the performance impact. 
\item We analyze the dynamic SBS on/off model and prove that sublinear regret is impossible for arbitrary SBS on/off. To solve this challenging problem, we create a novel strategy set, called \textit{ranking expert}, that is used in conjunction with a BREW-type solution with respect to expert advice. The novelty of the expert construction is that \textit{it simultaneously takes into account both the handover cost and the availability of SBS}. The regret upper bound with respect to the best expert advice is  proved.
\end{itemize}

The rest of the paper is organized as follows. The system model is presented in Section \ref{sec:model}. Section \ref{sec:learning} discusses the non-stochastic learning approach for mobility management, including the BREW algorithm in Section \ref{sec:BREW}, regret analysis in \ref{sec:reg_BREW}, robustness in Section \ref{sec:extBREW}, and performance analysis of the industry solutions in Section \ref{sec:3gpp}. Dynamic SBS presence is studied in Section \ref{sec:onoff}. Simulation results are presented in Section \ref{sec:sim}. Finally, Section \ref{sec:conc} concludes the paper.

\section{System Model} 
\label{sec:model}

\subsection{Network Model}
\label{sec:net_model}

An ultra-dense cellular network with $N$ small base stations (SBS) and $M$ user equipments (UE) is considered in this work. We denote the SBS set as $\mathcal{N}_{\mathrm{SBS}}=\{1, \cdots, N\}$. We are mostly concerned with stationary or slow moving UEs, which represents a typical indoor scenario where about 80\% of the total network traffic occurs \cite{NSN:2011}. A representative UE in UDN may have multiple SBSs as the potential serving cell, but needs to choose only one serving SBS. In other words, advanced technologies that allow for multiple serving cells are not considered.  One exemplary system is illustrated in Fig.~\ref{fig:FHO}, where UE $1$ may discover up to $6$ candidate SBSs in its neighborhood, possibly with very similar signal strength or load conditions. The mobility management system makes decisions on which UE is idly camped on (idle mode mobility) or actively served by (connected mode mobility) which SBS at any given time.   

\begin{figure}[htb]
    \centering
    \centerline{\includegraphics[width=0.4\textwidth]{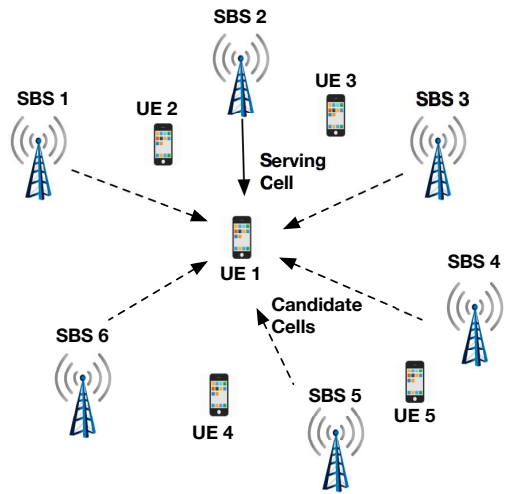}}
    \caption{An illustration of mobility management in UDN.}
    \label{fig:FHO}
\end{figure}

The mobility decision is traditionally made at the SBS (e.g., X2 handover in LTE) or Evolved Packet Core (e.g., S1 handover at the Mobility Management Entity in LTE). Recently, there has been an emerging trend of designing user-centric mobility management, particularly for the future 5G standard \cite{Boccardi:14}. In this work, we consider user-centric mobility management and let the UE make mobility decisions. We assume that mobility management is operated in a synchronous time-slotted fashion.  It is worth noting that we do not make any assumption on whether the candidate SBSs are operating in the same channel or different channels, as our work applies to both of these deployments.

The sequence of operations within each slot can be illustrated in Fig~\ref{fig:sys}. Specifically, at the beginning of a slot, the UE chooses its serving SBS and starts a downlink data transmission with the paired SBS. Upon the completion of the time slot, the UE can observe the total energy consumed over this slot, and the operation repeats in the next slot. 

\begin{figure}[htb]
    \centering
    \centerline{\includegraphics[width=0.8\textwidth]{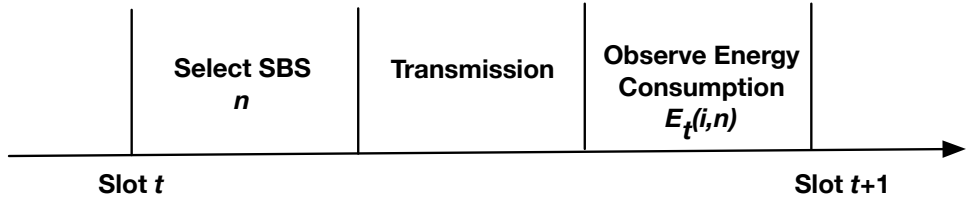}}
    \caption{The mobility management operation for UE $i$ in time slot $t$.}
    \label{fig:sys}
\end{figure}

\subsection{Non-stochastic Energy Consumption Model}
\label{sec:EC_model}

We study the mobility problem with minimizing long-term energy consumption as the system design objective, and adopt a \textit{non-stochastic} modelling of the energy consumption of each SBS.  Specifically, SBS $n$, $1 \leq n \leq N$ incurs a total energy consumption $E_{t}(i,n)$ if it serves UE $i$ at time slot $t$, $1 \leq t \leq T$. It is assumed that $E_{t}(i,n) \in [E_{\textup{min}},  E_{\textup{max}}]$. In this work, we make no statistical assumptions about the nature of the process that generates $E_{t}(i,n)$. In fact, we allow $\{E_{t}(i,n) \}$ to be any arbitrary sequence for any $n$ and any $i$. This is a fundamental difference to the stochastic MAB based mobility solutions \cite{Simsek:15,Capdevielle:12}. Note that $E_{t}(i,n)$ may include additional energy consumptions if the UE decides to switch from one SBS to another. The set of energy consumptions $\{E_{t}(i,n), n=1, \cdots, N; t=1, \cdots, T\}$ is unknown to UE $i$. We are interested in finding a SBS selection sequence $\{a_{i,t}, t=1, \cdots, T\}$ for UE $i$ that minimizes the total energy consumption over $T$ slots: $\sum_{t=1}^{T} E_t(i,a_{i,t}) $.

\section{BREW: a non-stochastic mobility management algorithm}
\label{sec:learning}

\subsection{Problem Formulation with Handover Cost}
\label{sec:mot}

We take a representative UE and drop the UE index $i$ from the notation. In the non-stochastic multi-armed bandit model, each arm $n$ corresponds to a SBS for which there exists an arbitrary sequence of energy consumptions up to time $T$ if the UE is served by this SBS. Let $a_t$ denote the SBS selected by the UE at time $t$. It is assumed that after time slot $t$, the UE only knows the energy consumptions $E_{1}(a_{1}), \cdots, E_{t}(a_{t})$ of the previously selected SBS $a_1, \cdots, a_t$. In other words, the UE does not gain any knowledge about the SBSs which it does not choose. Thus, the UE's mobility algorithm can be stated as selecting a sequence $a_1, \cdots, a_T$ where $a_t$ is a mapping from the previous actions\footnote{An SBS selection is also referred to as an \textit{action}.} and the corresponding observed energy consumptions from time $1$ to $t-1$ to the selection of a SBS at $t$. Note that the knowledge of past SBS activities can be practically enabled by leveraging the ``UE History Information'' element in the 3GPP LTE specs  \cite{3gpp.36.413}.

At each time slot $t=1, 2, \cdots, T$, the UE chooses SBS $a_t$ from $\mathcal{N}_{\mathrm{SBS}}$, and then observes an energy consumption $E_{t}(a_t)$ for data transmission, which is sent to the UE as feedback from the SBS. For an arbitrary sequence of energy consumptions $\{E_{t}(a_t)\}$ and for any $T>0$, we denote
\begin{equation}
E_{\mathbf{a},T,1} \doteq \sum_{t=1}^{T} E_{t}(a_t)
\end{equation}
as the total energy consumption without considering any handover cost, at time $T$ of policy $\mathbf{a}$. Note that $E_{\mathbf{a},T,1}$ captures the total energy consumption up to $T$,  corresponding to the service the UE receives from its (possibly varying) serving SBS. Clearly, $E_{\mathbf{a},T,1}$ depends on the arbitrary energy consumption sequences at each SBS as well as the UE actions. We refer to $E_{\mathbf{a},T,1}$ as the total \textit{service energy consumption}. 

In practice, switching from one SBS to another incurs additional cost, and frequent switching incurs large energy consumption that is not captured by the service energy consumption $E_{\mathbf{a},T,1}$. To address this issue, we explicitly add additional energy consumption whenever a handover occurs, and thus force the optimal solution to minimize frequent handovers. For simplicity, we assume that a homogeneous energy consumption $E_s \geq 0$ is incurred whenever a UE is handed over from one SBS to another. This cost includes all energy consumptions that are associated with handovers, such as sending additional overhead signals and forwarding UE packets. The total \textit{handover energy consumption} can be computed as
\begin{equation}
\label{eqn:HOcost}
E_{\mathbf{a},T,2} = E_s \sum_{n=1}^{N} \sum_{t=2}^{T} \mathds{1}_{\{a_t=n, a_{t-1}\neq n\}}
\end{equation}
where $\mathds{1}_A$ is the indicator function for event $A$. As opposed to the service energy consumption, the handover energy consumption only depends on the UE action $\mathbf{a}$.

Finally, the total energy consumption over $T$ slots with handover cost can be written as
\begin{equation}
\label{eqn:Wcost}
E_{\mathbf{a},T} = E_{\mathbf{a},T,1} + E_{\mathbf{a},T,2},
\end{equation}
and we are interested in finding a mobility management policy that minimizes $\mathbb{E}[E_{\mathbf{a},T}]$. It is worth noting that by including the handover cost (\ref{eqn:HOcost}) in the total energy consumption (\ref{eqn:Wcost}), a good handover algorithm not only has to balance the tradeoff between exploitation and exploration, but also needs to minimize the number of occurrences that the UE changes SBS associations. Hence, the FHO problem is implicitly solved when the UE total energy consumption is minimized.

\subsection{The BREW Algorithm}
\label{sec:BREW}

To simplify the analysis, we assume without loss of generality that $E_{\min}=0$, and both $E_{\max}$ and $E_s$ are normalized as $E_{\max} + E_s = 1$.  Thus, if we re-write the total energy consumption of selecting SBS $a$ at time slot $t$ as $\tilde{E}_t(a) := E_t(a) + E_s \mathds{1}_{\{a_t \neq a_{t-1}, t>1 \}} $, we have $\tilde{E}_a(t) \in [0, 1]$. We also assume that the energy consumption for SBS $a$ is arbitrary but oblivious. In practice, this assumption is valid when the service energy consumption $E_t(a)$ of UE selecting SBS $a$ at time $t$ only depends on the current state of SBS $a$, such as its user load and traffic load, minimum transmit power to satisfy UE's QoS, etc. In other words, we do not consider the case that the SBS intelligently manipulates its service energy consumption to counter the UE mobility policy it learns from the past.

\begin{algorithm} 
\DontPrintSemicolon
\SetKwInOut{Input}{Input}
\SetKwInOut{Output}{Output}
\SetKwInput{Init}{Initialize}
\Input{A non-increasing sequence $\{ \gamma_l \}_{l \in \mathbb{N} }$, $\tau \in \mathbb{N}_{+}$}
\Init{$p_a(l) = 1/N$ and $\hat{L}_{0}(a) = 0$ for all $a \in {\cal N}_{\text{SBS}}$}

 \While{$l \geq 1$}{
 	Select SBS $a(l)$ randomly according to the probabilities $\{p_a(l)\}$, $a \in {\cal N}_{\text{SBS}}$ \;
	Keep UE on SBS $a(l)$ for the next $\tau$ time slots: $(l-1) \tau + 1, \ldots, l \tau$ \;
	Observe total energy consumption $\{ \tilde{E}_{t}(a(l)) \}_{t=(l-1) \tau + 1}^{ l \tau}$, possibly including a one-time handover energy consumption $E_s$ at $(l-1) \tau + 1$\;
	Calculate the average energy consumption incurred in batch $l$:
		\begin{equation}
		\bar{E}_{t}(a(l)) = \frac{1}{\tau} \sum_{t=(l-1) \tau + 1}^{ l \tau}  \tilde{E}_{t}(a_t)
		     \label{eqn:alg1}
		\end{equation} 
	Calculate the estimated energy consumption of each $a \in {\cal N}_{\text{SBS}}$ in the batch
		\begin{equation}
		\hat{E}_{l}(a) = \frac{ \bar{E}_{t}(a(l)) }{ p_{a}(l)  } \mathds{1}_{\{a = a(l)\}}   \label{eqn:alg2}
		\end{equation}	
	Update the cumulative estimated energy consumption of each $a \in {\cal N}_{\text{SBS}}$
		\begin{equation}
		\hat{L}_{l}(a) = \hat{L}_{l-1}(a) + \hat{E}_{l}(a)   \label{eqn:alg3}
		\end{equation}
	For $a \in {\cal N}_{\text{SBS}}$ set
		\begin{equation}
		p_a(l+ 1) 
		= \frac{ \exp (-\gamma_l \hat{L}_{l}(a) )   }   
		{ \sum_{a' \in {\cal N}_{\text{SBS}}}  \exp (-\gamma_l \hat{L}_{l}(a') )  }    \label{eqn:alg4}   
		\end{equation}
	$l = l+1$ \;
 }
 \caption{The  BREW mobility management algorithm.}
 \label{alg:BREW}
\end{algorithm}

The proposed Batched Randomization with Exponential Weighting (BREW) solution is given in Algorithm \ref{alg:BREW}. As the name suggests, it is a batched extension of the exponential weighting algorithm such as the celebrated EXP3 \cite{auer2}. Note that in Algorithm \ref{alg:BREW}, the EXP3 component is an adaptation of the algorithm originally proposed in \cite{auer2}, which uses a slightly different weighing scheme and works with loss functions instead of reward functions \cite{bubeck2010jeux} to get rid of the uniform mixture term in the probabilistic decision rule of the EXP3 in \cite{auer2}. We highlight several key design considerations. Firstly, because the energy consumption of each SBS can be generated arbitrarily, it is easy to show that for any \textit{deterministic} mobility solution, there exist sequences of energy consumption that make the solution highly sub-optimal. In other words, no fixed algorithm can guarantee a small performance degradation against \textit{all} possible energy consumption sequences. Hence, for the non-stochastic mobility problem, we introduce \textit{randomization} in the proposed algorithm to avoid being stuck in a worst-case energy consumption. This is done by selecting SBS based on a probability distribution over $N$ SBSs. 

Subsequently, a natural question is what type of randomization one should introduce to achieve good energy consumption performance. We note that our mobility management problem can be viewed as a special case of \textit{optimal sequential decision for individual sequences} \cite{Merhav:02}, for which \textit{exponential weighting} is a fundamental tool. The proposed algorithm uses exponential weighting to construct and update the probability for choosing SBS, as shown in (\ref{eqn:alg4}).

Finally, in order to address the FHO problem and avoid incurring large accumulated handover energy consumption, we need to ``explore in bulk''. This is done by grouping time slots into batches and not switching within each batch. What separates the operations within a batch from outside is that the UE does not observe energy consumption on a per-slot basis and does not need to update the internal state. In general, BREW works as if the UE is unaware that a batch has happened as opposed to one time slot. At the end of the batch, though, the UE can receive a one-time energy consumption feedback, which is the average energy during the batch as shown in equation (\ref{eqn:alg1}). The choice of the batch length plays a critical role in the overall performance -- if it is too large, one may get the benefit of having little loss from the handover energy consumption, but also may stuck at a sub-optimal SBS for a long time, and vice versa.  The BREW algorithm uses a parameter $\tau$ that determines the batch length.

\subsection{Finite-Time Performance Analysis}
\label{sec:reg_BREW}

To evaluate the performance of the proposed BREW mobility solution, we adopt a \textit{regret} formulation that is commonly used in multi-armed bandit theory \cite{Bubeck:12}. Specifically, we compare the energy consumption of the BREW algorithm with a ``genie-aided'' solution where UE chooses the SBS which has the minimum total energy consumption over $T$ slots, i.e., chooses the best SBS with the minimum $E_{\mathbf{a},T,1}$ up to time $T$ and incurs no handover cost $E_{\mathbf{a},T,2} = 0$. Our goal is to characterize the energy consumption regret for any finite time $T$. The smaller this regret is, the better the solution. 

Formally, we define
\begin{equation}
\label{eqn:genie}
E_{\textup{best}} \doteq \min_{n \in {\cal N}_{\text{SBS}} }  \sum_{t=1}^{T} E_{t}(n)
\end{equation}
as the energy consumption of the single best SBS at time $T$. Then, the performance of any UE mobility solution $\mathbf{a}$ can be measured against the genie-aided optimal policy (\ref{eqn:genie}) in expectation. We formally define the regret of a UE mobility solution $\mathbf{a}$ as:
\begin{align}
R_{\mathbf{a}}(T) 
:= \sum_{t=1}^T \bb{E} \left[ E_t( a_t ) + E_s \mathds{1}_{\{a_t \neq a_{t-1}, t>1\}}\right] - E_{\textup{best}},
\label{eqn:fixedregret}
\end{align}
which is the difference between the total energy consumption of the learning algorithm and the total energy consumption of the best fixed action by time $T$. Here $a_t$ denotes the action chosen by the UE at time slot $t$ and the expectation is taken over the randomization of the UE's algorithm.

Note that $R_{\mathbf{a}}(T)$ is a non-decreasing function of $T$. For any mobility algorithm to be able to learn effectively, $R_{\mathbf{a}}(T)$  has to grow \textit{sublinearly} with $T$. In this way, one has $\lim_{T \rightarrow \infty} R_{\mathbf{a}}(T)/T = 0$, indicating that asymptotically the algorithm has no performance loss against the genie-aided solution. For the BREW mobility solution, we have the following theorem that upper bounds its regret for any finite time $T$.

\begin{theorem} 
\label{thm:SwitchRegret} 
For a given time horizon $T$, when BREW (Algorithm \ref{alg:BREW}) runs with $\gamma_l = \sqrt{ (2 \log N)/ (l N) }$ and batch size $\tau = \lceil B_N T^{1/3} \rceil$, where $B_N = ( 4.5 N \log N )^{-1/3}$, its regret is bounded by
\begin{equation}
\label{eqn:BREWreg}
R_{\mathbf{a}}(T) \leq   2 B^{-1}_N  T^{2/3} + \left( B_N + B^{-2}_N  \right) T^{1/3} + 1.
\end{equation}
\end{theorem}

\begin{proof}
See Appendix~\ref{apd:thm1}.
\end{proof}

Theorem \ref{thm:SwitchRegret} provides a sublinear regret bound for BREW that guarantees the long-term optimal performance. Moreover, the bound in Theorem \ref{thm:SwitchRegret} applies to any finite time $T$, and can be used to characterize how fast the algorithm converges to the optimal action. Specifically, the total energy consumption of the UE that uses the BREW algorithm will approach, on average,  the total energy consumption of the best fixed action $a \in {\cal N}_{\text{SBS}}$ at a rate no slower than $O(T^{-1/3})$. Although the BREW algorithm works for any non-increasing sequence of $\gamma_l$, the particular choice of $\gamma_l = \sqrt{ (2 \log N)/ (l N) }$ gives a guaranteed upper bound of the regret in (\ref{eqn:BREWreg}). 

Another important remark is that the choice of $\tau$ given in Theorem \ref{thm:SwitchRegret} is optimal in terms of the time order of the regret, which is $O(T^{2/3})$. It is shown in \cite{dekel2014bandits} that for any learning algorithm, there exists a loss sequence under which that learning algorithm suffers $\tilde{\Omega}(T^{2/3})$ regret. Hence, our choice of $\tau$ in Theorem \ref{thm:SwitchRegret} results in a regret upper bound that matches the lower bound of $O(T^{2/3})$, proving its optimality in terms of the time order. A smaller value for $\tau$ will make BREW incur a higher cost due to over-switching, while a larger value for $\tau$ will make the algorithm incur a higher cost due to staying too long on a suboptimal action. 

\subsection{Robustness Analysis}
\label{sec:extBREW}

The development of BREW for energy-efficient mobility management captures the basic characteristics of the user mobility model in UDN, particularly with the inclusion of handover cost as well as the non-stochastic nature of the approach. However, in a real-world deployment, robustness issues often arise, such as delayed or missing feedback of the energy consumption. Whether the BREW algorithm can handle these problems and its impact on the total energy consumption is essential for its practical deployment.

\subsubsection{Delayed Feedback}
\label{sec:delayed}

Delayed feedback constantly happens in practice. For example, in most cellular standards, sending the feedback to UE may be delayed because it has to wait for the frames that are dedicated to sending control and signalling packets.

The next theorem provides a regret bound  when the UE receives the information about the energy consumption with a delay of $d$ time slots. 

\begin{theorem} 
\label{thm:DelayRegret} 
Consider the UE-SBS mobility operation in Fig.~\ref{fig:sys} where a feedback at time slot $t$ is $d$-delayed energy consumption $E_{t-d+1}(i,n)$, $d \geq 1$.
For a given time horizon $T$, when BREW (Algorithm \ref{alg:BREW}) runs with $\gamma_l = \sqrt{ (2 \log N)/ (l N) }$ and batch size $\tau = \lceil B_N T^{1/3} \rceil$, where $B_N = ( 4.5 N \log N )^{-1/3}$, its regret is bounded by
\begin{equation}
R_{\mathbf{a}}(T) \leq  
(d+1) B^{-1}_N  T^{2/3} + \left( B_N + B^{-2}_N  \right) T^{1/3} + 1 . 
\label{eqn:DelayRegret}
\end{equation}
\end{theorem}

\begin{proof}
See Appendix~\ref{apd:thm2}.
\end{proof}

The importance of Theorem \ref{thm:DelayRegret}  is that it preserves the order-optimality of the BREW algorithm in the presence of delayed feedback. In fact, comparing (\ref{eqn:DelayRegret}) to (\ref{eqn:BREWreg}), we can see that only the coefficient before $T^{2/3}$ is increased due to the delayed feedback. Moreover, Theorem \ref{thm:DelayRegret} holds when $d <  B_N T^{1/3}$, implying that as long as the delay is sublinear in time with a small enough time exponent, the regret can be guaranteed to be sublinear in time. Finally, we note that in practice, the delay of feedback may be time-varying, and Theorem \ref{thm:DelayRegret} can be applied to varying delays where $d$ is chosen as the maximum feedback delay.

\subsubsection{Missing Feedback}
\label{sec:missing}

Another practical issue with respect to UE observing the energy consumption is that such feedback from the SBS to UE may be entirely missing. This can happen when the downlink transmission that carriers the feedback information is not correctly received at the UE.

We take a probabilistic approach when considering the missing feedback problem. Specifically, we assume that at each time slot $t$, the energy consumption feedback may be missing with probability $p$, which can be set, e.g., as the packet loss rate for the transmission. 

Algorithm~\ref{alg:BREWmiss} is the modified BREW that can handle the missing feedback. The idea behind Algorithm~\ref{alg:BREWmiss} is to make the estimated energy consumption $\hat{E}_l(a)$ an \textit{unbiased} estimate of the actual average energy consumption in batch $l$. Therefore, we normalize this quantity by dividing it with the probability that $\zeta$ feedbacks are observed in batch $l$. With this approach, the contribution of rare events (unlikely feedback sequences) to the cumulative estimated energy consumption is magnified. This idea is widely used in the design of exponential weighing algorithms and their variants \cite{Merhav:02}.

\begin{algorithm} 
\DontPrintSemicolon
\SetKwInOut{Input}{Input}
\SetKwInOut{Output}{Output}
\SetKwInput{Init}{Initialize}
\Input{A non-increasing sequence $\{ \gamma_l \}_{l \in \mathbb{N} }$, $\tau \in \mathbb{N}_{+}$; probability $p$ that feedback will be missing in each time slot.}
\Init{$p_a(1) = 1/N$ and $\hat{L}_{0}(a) = 0$ for all $a \in {\cal N}_{\text{SBS}}$}

 \While{$l \geq 1$}{
 	Select SBS $a(l)$ randomly according to the probabilities $\{p_a(l)\}$, $a \in {\cal N}_{\text{SBS}}$ \;	
	Keep UE on SBS $a(l)$ for the next $\tau$ time slots: $(l-1) \tau + 1, \ldots, l \tau$ \;
	Let ${\cal O}(l)$ be the set of time slots in $\{(l-1) \tau + 1, \ldots, l \tau \}$ for which the feedback is observed \;
	Observe total energy consumption $\{ \tilde{E}_{t}(a(l)) \}_{t \in {\cal O}(l)}$, possibly including a one-time handover energy consumption $E_s$ at $(l-1) \tau + 1$\;
	Calculate the average energy consumption incurred in batch $l$:
		\begin{equation}
		\bar{E}_{l} = \frac{1}{|{\cal O}(l)|} \sum_{ t \in {\cal O}(l) }  \tilde{E}_{t}( a(l) ) 
			     \label{eqn:algBREWmiss1}
		\end{equation} 
	Calculate the estimated energy consumption of each $a \in {\cal N}_{\text{SBS}}$ in the batch
		\begin{equation}
		\hat{E}_{l}(a) = \frac{ \bar{E}_{l} }{ p_{a}(l) {\tau \choose \zeta} (1-p)^{\zeta} p^{\tau - \zeta} } \mathds{1}_{\{a = a(l)\}} \mathds{1}_{ \{|{\cal O}(l)| = \zeta \}} 
			\label{eqn:algBREWmiss2}
		\end{equation}	
	Update the cumulative estimated energy consumption of each $a \in {\cal N}_{\text{SBS}}$
		\begin{equation}
		\hat{L}_{l}(a) = \hat{L}_{l-1}(a) + \hat{E}_{l}(a)   \label{eqn:algBREWmiss3}
		\end{equation}
	For $a \in {\cal N}_{\text{SBS}}$ set
		\begin{equation}
		p_a(l+ 1) 
		= \frac{ \exp (-\gamma_l \hat{L}_{l}(a) )   }   
		{ \sum_{a' \in {\cal N}_{\text{SBS}}}  \exp (-\gamma_l \hat{L}_{l}(a') )  }    \label{eqn:algBREWmiss4}   
		\end{equation}
	$l = l+1$ \;
 }
 \caption{The modified BREW algorithm for missing energy consumption feedback with probability $p$.}
 \label{alg:BREWmiss}
\end{algorithm}

Unfortunately, we are not able to prove a regret bound that grows sublinearly in time with respect to the best fixed action. Intuitively, for large $T$ there will be approximately $pT$ time slots where no feedback is received. Since we are dealing with the non-stochastic bandit problem, the worst-case loss for these slots can be linear with $T$ even when $T$ is large. This is a significant difference to the stochastic case where when $T$ is large, the concentration property can guarantee a sublinear loss in time. 

\subsection{Analyze the 3GPP Mobility Protocols Under the Online Learning Framework}
\label{sec:3gpp}

The proposed BREW mobility algorithm and its variations are developed under an online learning framework. In this section, we put the existing industry mobility mechanism under the same framework and characterize its regret performance through the lens of MAB.

The handover mechanism defined in 3GPP centers on the UE measuring the signal quality from candidate SBSs. Such measurements are generally filtered at the UE to rule out outliers. The original 3GPP handover protocol chooses the SBS with the highest signal quality to serve the UE, and sticks with the choice until some performance metric (such as RSRP or RSRQ when LTE is used \cite{SesiaLTE}) drops below a threshold, at which time the UE measures all candidate SBSs again and hands over to the best neighbor. 

The original protocol is optimization-based and has no restrictions on how frequent handovers can happen. Recognizing that FHO can happen when the network density is high, there have been some proposals in 3GPP to modify the handover parameters when frequent handovers are observed. The general principle is to first determine whether a FHO problem has happened, typically by counting the number of handovers within a sliding time window. If it is determined that there are too many handovers, mobility parameters such as hysteresis margin and time-to-trigger are modified to ``slow down'' future handovers, thus avoiding FHO and making the current serving SBS more ``sticky''.

The original handover mechanism and its variation can be viewed as a myopic rule that is both \textit{greedy} (always select the SBS that is immediately the best) and \textit{conservative} (only take actions when the selected SBS becomes bad enough). The following proposition shows the sub-optimality of these approaches. 
\begin{prop}
\label{prop:3gpp1}
Under both stochastic and non-stochastic MAB models for SBS energy consumption, the original and enhanced 3GPP handover protocols describe above achieve an asymptotic regret of $O(T)$.
\end{prop}

\begin{proof}
See Appendix~\ref{apd:3gpp1}.
\end{proof}

Proposition~\ref{prop:3gpp1} proves a linear regret with respect to time $T$ for the 3GPP handover solutions. Two important remarks are in place. Firstly, Proposition~\ref{prop:3gpp1} is a strong result in the sense that it  is proved for both stochastic and non-stochastic energy consumption, meaning that linear regret is inevitable regardless of the adopted model. Secondly, the sub-optimality is shown without considering the handover cost. In other words, existing industry mechanisms cannot converge to the best SBS even when $E_s =0$. A non-zero $E_s$ will further deteriorate the regret performance. Detailed numerical comparisons will be made in Section~\ref{sec:sim}.

\section{Dynamic SBS Presence}
\label{sec:onoff}

In this section we consider energy efficient mobility management for SBSs with dynamic presence. Notably, this is a new problem that arises with the increase of user-deployed SBSs. For both enterprise and residential small cell deployment, SBSs can be turned on and off by users, thus creating problems for mobility management. In particular, such on/off behaivor would disrupt the learning process. To capture this uncontrolled user behavior, we consider the following generic SBS on/off model. At each time slot $t$, a subset of SBSs, chosen \textit{arbitrarily}, can be turned off and hence cannot serve the UE. As we will see, this problem is significantly harder. There are known results in the literature of \textit{stochastic} multi-armed bandits with appearing and disappearing arms \cite{Chakrabarti:NIPS2008,Bnaya:13}, but the theoretical structure of these solutions are very different from the \textit{non-stochastic} problem in this paper. Consequently, we have to develop new results in the non-stochastic bandit theory and design robust mobility management solutions.

The set of SBSs available at time slot $t$ is denoted by ${\cal N}_t \subset {\cal N}_{ \text{SBS} }$.  An SBS in ${\cal N}_t$ is called an {\em active} SBS, while an SBS in ${\cal N}_{ \text{SBS} } - {\cal N}_t := \{ n: n \in {\cal N}_{ \text{SBS} }, n \notin {\cal N}_t  \}$ is called an {\em inactive} SBS. We require that the UE only selects from the active SBS set ${\cal N}_t$ at time slot $t$.

We first give an impossibility result regarding achieving the optimal performance asymptotically. Theorem \ref{thm:linearregret} shows that in general it is impossible to obtain sublinear regret when ${\cal N}_t$ changes in an arbitrary way.

\begin{theorem} 
\label{thm:linearregret}
Assume that $\{ {\cal N}_t \}_{t=1}^T$ is generated by an adaptive adversary which selects ${\cal N}_t$ based on the action chosen by the learning algorithm at time slot $t-1$, and 
$\{ E_t \}_{t=1}^T$ is generated by an oblivious adversary. Then, for $E_s \geq E_{\max} + 1/(N-1)$, no learning algorithm can guarantee sublinear regret in time. 
\end{theorem}

\begin{proof}
See Appendix~\ref{apd:thm3}.
\end{proof}

In light of the impossibility result in Theorem \ref{thm:linearregret}, the pursue of good energy consumption performance in the presence of dynamic SBS on/off is only viable when the generation of ${\cal N}_t$ is constrained. In the following discussion, we will focus on an \textit{i.i.d. SBS activity model}, where SBS $a$ is present at time slot $t$ with probability $p_a$ independently from other time slots and other SBSs. We emphasize that the i.i.d.  activity model generally represents a case that is worse than practice, where the SBS on/off introduced by end-users has some memory. The correlation over time can be exploited by a learning algorithm to achieve better mobility. We focus on the i.i.d.  activity model because it presents a more challenging SBS dynamic, and the resulting algorithms and regret analysis can serve as a guideline to the real-world SBS on/off performance. The proposed algorithms for the i.i.d.  model can be applied to other SBS activity models, such as the Markov model. Furthermore, note that both i.i.d. and Markov models are widely used in stochastic multi-armed bandit, but in our paper they are used for modelling the SBS on/off activities, not the reward distribution. 

In order to address the SBS dynamics, we follow the general principle of \textit{prediction with expert advice} \cite{Cesa:06}. In this setting, we assume that there is a set of experts which recommend the SBS for the UE to select, based on the past sequence of selections and energy consumption feedback to the UE. There is no assumption on the way these experts compute their predictions, and the only information UE receives from these experts is the advice. Then, we will bound the regret of the UE with respect to the best of these experts.

\subsection{Ranking Expert}
\label{sec:ranking}

The first proposed algorithm utilizes a concept called ``ranking expert'' \cite{kleinberg2010sleeping} and the algorithm consists of two key elements. Firstly, we will need an efficient expert selection procedure that chooses the action based on all expert advices. In order to achieve low regret, a modified EXP4 procedure from \cite{auer2} is used to select the expert at each time slot. The second component is how to construct expert advice, where we use \textit{ranking} to sort the possible actions at each expert. The overall Ranking Expert (RE) algorithm is given in Algorithm~\ref{alg:RE}. 

\begin{algorithm} 
\DontPrintSemicolon
\SetKwInput{Input}{Input}
\SetKwInput{Init}{Initialize}
\Input{A non-increasing sequence $\{ \gamma_t \}_{t \in \mathbb{N} }$; $\tilde{{\cal E}} = {\cal E} \cup \text{Unif}$ (see Definition~\ref{def:unif})}
\Init{$q_e(1) = 1/ |\tilde{{\cal E}} |$ and $\hat{L}_0(e) = 0$ for all $e \in \tilde{{\cal E}}$, $a_0 = \text{Rand}( {\cal N}_{\text{SBS}} )$}

 \While{$t \geq 1$}{
 	Observe $a_{t-1}$ and ${\cal A}_t$ \;
 	Get ranking expert advices $\{ \bs{\delta}^e(t) \}_{e \in \tilde{{\cal E}}}$\;
 	Set 
		\begin{equation}
			p_a(t) = \sum_{e \in \tilde{{\cal E}}} q_e(t) \delta^e_a(t)
		\end{equation}
	Select $a_t$ randomly according to the probabilities $p_a(t)$, $a \in {\cal N}_{\text{SBS}}$ \;
	Observe energy consumption $ \tilde{E}_{t}(a_t) \in [0,1]$ \;
	Set 
		\begin{equation}
   		\hat{X}_{t}(a) = \left\{
                		\begin{array}{ll}
                  		\tilde{E}_{t}(a) / p_a(t) &\text{ if } a = a_t \\
                  		0 &\text{ otherwise} \\
                		\end{array}
              	\right. , \text{ for } a \in {\cal N}_{\text{SBS}}
		\end{equation}
	Update the cumulative estimated energy consumption of each ranking expert $e \in \tilde{{\cal E}}$
		\begin{equation}
			\hat{L}_{t}(e) = \hat{L}_{t-1}(e)  +  \hat{X}_{t}(a_t) \delta^e_{a_t}(t)   
		\end{equation}
		
	For $a \in {\cal N}_{\text{SBS}}$ set
		\begin{equation}
		q_e(t+ 1)  = \frac{ \exp (-\gamma_t \hat{L}_{t}(e) )  }{ \sum_{e' \in \tilde{{\cal E}} }  \exp ( -\gamma_t \hat{L}_{t}(e') ) }  
		\end{equation}
	$t = t+1$ \;
 }
 \caption{The Ranking Expert (RE) mobility management algorithm.}
 \label{alg:RE}
\end{algorithm}

Let $\bs{\delta}^e = ( \delta^e_1, \ldots, \delta^e_N  )$ denote the action choice vector of expert $e$, where $\delta^e_a = 1$ denotes the event that expert $e$ recommends SBS $a$, and $\delta^e_a = 0$ denotes the event that expert $e$ does not recommend SBS $a$. It is assumed that each expert recommends only one SBS in ${\cal N}_{\text{SBS}}$. Since we have both handover cost and dynamic SBS activity, we consider experts whose recommendation strategy at time $t$ depends on both $a_{t-1}$ and ${\cal N}_t$. This is a critical step, because otherwise the handover energy consumption will not be considered by the pool of experts. We specifically focus on {\em ranking experts} whose preference over the set of SBSs is given by a previous action-dependent ranking. 

\begin{defi}
\label{def:3}
An expert is called a {\em ranking expert} if for each previous action $a \in {\cal N}_{\text{SBS}}$, expert $e$ has a ranking over ${\cal N}_{\text{SBS}}$ given by $\sigma_{e,a}$. Let ${\cal E}$ denote the set of all possible ranking experts, with size $N_E := |{\cal E}|$. 
\end{defi}

One benefit of considering ranking experts is that the ranking can be performed on the entire set of SBS ${\cal N}_{\text{SBS}}$, but the recommendation can take into account of the SBSs that are turned off, by removing them from the ordered set. Specifically, given ${\cal N}_t$, expert $e \in {\cal E}$ recommends the action with the highest rank in ${\cal N}_t$, which is denoted by $\sigma_{e,a}({\cal N}_t)$. 
 
Different from the definition of regret in (\ref{eqn:fixedregret}), which is with respect to the best fixed \textit{action}, we define the regret of the UE in this section with respect to the best fixed \textit{expert}. For expert $e \in {\cal E}$, let $a^e_t$ denote the action recommended at time $t$, and $a_t$ denote the random variable that represents the action chosen by the UE at time $t$. Since $e$ is a ranking expert, $a^e_t$ depends on ${\cal N}_t$ and $a^e_{t-1}$.
We define the regret with respect to the best expert from a pool of experts ${\cal E}$ as 
\begin{equation}
R_{\mathbf{a}} ({\cal E}, T) 
 :=   \sum_{t=1}^T \bb{E} \left[ E_t( a_t ) + E_s \mathds{1}_{\{a_t \neq a_{t-1}\} }\right]  
 -   \bb{E} \left[ \min_{e \in {\cal E} } \sum_{t=1}^T
\left[ E_t( a^e_t ) + E_s \mathds{1}_{ \{a^e_t \neq a^e_{t-1} \}} \right]  
\right]
\end{equation}

Due to a technicality, we need to introduce a {\em uniform expert} (denoted by Unif) in order to bound the regret \cite{auer2}. It is  defined as the following.
\begin{defi}
\label{def:unif}
An expert is called a {\em uniform expert} ('Unif') if it recommends action $a \in {\cal N}_{\text{SBS}}$ with probability $1/N$, regardless of whether $a \in {\cal N}_{t}$.
\end{defi}
We denote the extended pool of experts which includes all the experts in ${\cal E}$ and the uniform expert as $\tilde{{\cal E}}$.  Instead of having a deterministic SBS selection rule like the other experts, the uniform expert selects its action at time slot $t$ according to the uniform distribution on ${\cal N}_{\text{SBS}}$ independently from past action and ${\cal N}_t$. Hence the uniform expert can recommend actions that are not in ${\cal N}_t$, which is now allowed. To address this issue, we assume that the UE randomly selects one of the actions in ${\cal N}_t$ if Algorithm~\ref{alg:RE} recommends an action that is not in ${\cal N}_t$. 

The theorem below gives a regret bound with respect to the best expert from the pool of experts defined above. 

\begin{theorem} 
\label{thm:ranking}
Assume that the UE uses the RE algorithm with the pool of ranking experts ${\cal E}$ defined in Definition~\ref{def:3} with $\gamma_t =  \sqrt{ \frac{ \log (1+N_E)}{ t N} }$.  We have 
\begin{equation}
R_{\mathbf{a}} ({\cal E}, T) \leq 2  N \sqrt{T N  \log N}.
\end{equation}
\end{theorem}

\begin{proof}
See Appendix~\ref{apd:thm4}.
\end{proof}

We have the following two remarks regarding the RE algorithm and its regret analysis for dynamic SBS on/off. Firstly, when the SBS can be turned on and off, the definition of the regret with respect to the best fixed SBS is no longer a strong definition, as any fixed SBS can be off for some slots. This is what motivated the regret definition with respect to the best fixed expert. Second, the ranking expert approach can simultaneously take care of handover cost (by letting the expert consider the previous action) and SBS on/off (by letting the expert recommend ranked SBS that is not off). 

\subsection{Contextual Ranking Expert}
\label{sec:contextual}

Although the RE algorithm defined in Algorithm~\ref{alg:RE} achieves sublinear regret with respect to $T$ and the regret bound given in Theorem \ref{thm:ranking} depends logarithmically on the number of ranking experts, one significant drawback is that maximum number of experts equals $N_E=(N!)^N$, which makes  practical implementations computationally challenging since the algorithm needs to keep a probability distribution over the set of experts, which is very large even for moderate $N$.

To address this practical issue, we re-visit the definition of ranking expert and  reduce the number of experts by \textit{defining ranking experts in the context of previous actions}. Consider the following contextual experts problem. Let ${\cal X}$ denote the context space. For a sequence of contexts $x_1, x_2, \ldots, x_T$, let $\tau_x \subset \{1, \ldots, T \}$ denote the set of time slots for which the context is $x \in {\cal X}$.  For the mobility management problem, we take the context at time $t$ to be the last action selected by the learning algorithm, i.e., $x_t = a_{t-1}$.\footnote{Note that in this definition the context is endogenously defined, i.e., it depends on the actions selected by the learning algorithm.}
Based on this, the contextual regret of the learning algorithm that works on a set of experts ${\cal E}$ is defined as  
\begin{equation}
R_{C}({\cal E}, T)  := 
 \sum_{b \in {\cal N}_\text{SBS} }   \bb{E} 
\Big[
 \sum_{t \in \tau_b }  \Big[  E_t( a_t ) + E_s \mathds{1}_{\{a_t \neq b \}} \Big] 
 - 
 \min_{e \in {\cal E} }  \sum_{t \in \tau_b } \Big[ E_t( a^e_t(b) ) + E_s \mathds{1}_{\{ a^e_t(b) \neq b\}}   \Big]  
\Big] 
\end{equation}
where $a^e_t(b)$ denotes the action chosen by expert $e$ at time $t$ based on context $b$ and the set of available actions ${\cal N}_t$, and the expectation is taken with respect to the randomization of the learning algorithm.

With the introduction of contextual experts, we can reduce the number of ranking experts exponentially by using a variant of the RE algorithm. Formally, we define the set of contextual ranking experts, which in contrast to the set of experts given in Definition \ref{def:3}, do not take into account the previous action when ranking the actions. 

\begin{defi}
\label{def:CEXP4}
An expert $e$ is called a {\em basic ranking expert} if it has a ranking over ${\cal N}_{\text{SBS}}$ given by $\sigma_{e}$. Let ${\cal E}$ denote the set of all possible basic ranking experts. 
\end{defi}

\begin{algorithm} 
\DontPrintSemicolon
\SetKwInput{Input}{Input}
\SetKwInput{Init}{Initialize}
\Input{A non-increasing sequence $\{ \gamma_t \}_{t \in \mathbb{N} }$; $\tilde{{\cal E}} = {\cal E} \cup \text{Unif}$}
\Init{ $q_{e,x}(1) = 1/ |\tilde{{\cal E}}| $, $\kappa(x) = 1$ and $\hat{L}_{0,x}(e) = 0$ for all $e \in \tilde{{\cal E}}$, $x \in {\cal N}_\text{SBS}$, $a_0 = \text{Rand}( {\cal N}_{\text{SBS}} )$}

 \While{$t \geq 1$}{
 	Observe $x_t = a_{t-1}$ and ${\cal A}_t$ \;
 	Get ranking expert advices $\{ \bs{\delta}^e(t) \}_{e \in \tilde{{\cal E}}}$ \;
 	Set 
		\begin{equation}
			p_a(t) = \sum_{e \in \tilde{{\cal E}}} q_{e,x_t}( \kappa(x_t)  ) \delta^e_a(t)   
		\end{equation}
	Select $a_t$ randomly according to the probabilities $p_a(t)$, $a \in {\cal N}_{\text{SBS}}$ \;
	Observe energy consumption $ \tilde{E}_{t}(a_t) \in [0,1]$ \;
	Set 
		\begin{equation}
   		\hat{X}_{t}(a) = \left\{
                		\begin{array}{ll}
                  		\tilde{E}_{t}(a) / p_a(t) &\text{ if } a = a_t \\
                  		0 &\text{ otherwise} \\
                		\end{array}
              	\right. , \text{ for } a \in {\cal N}_{\text{SBS}}
		\end{equation}
	Update the cumulative estimated energy consumption of each ranking expert $e \in {\cal E}$ for context $x_t$
		\begin{equation}
			\hat{L}_{\kappa(x_t),x_t}(e) = \hat{L}_{\kappa(x_t)-1,x_t}(e) +  \hat{X}_{t}(a_t) \delta^e_{a_t}(t) 
		\end{equation}
	For $a \in {\cal N}_{\text{SBS}}$ set
		\begin{equation}
		q_{e,x_t}( \kappa(x_t) + 1)  = \frac{ \exp (-\gamma_{\kappa(x_t) } \hat{L}_{\kappa(x_t) ,x_t}(e) ) }  { \sum_{e' \in {\cal E} }  \exp ( -\gamma_{\kappa(x_t) } \hat{L}_{\kappa(x_t) ,x_t}(e') )  } 
		\end{equation}
	Set $\kappa(x_t) = \kappa(x_t) + 1$ \;
	$t = t+1$ \;
 }
 \caption{The Contextual Ranking Expert (CRE) mobility management algorithm.}
 \label{alg:CRE}
\end{algorithm}

We are now in the position to propose a Contextual Ranking Expert (CRE) algorithm, which is given in Algorithm~\ref{alg:CRE}. The CRE algorithm uses the last action as the {\em context} and learns the best expert independently for each context. CRE runs a different instance of ranking experts for each context. It keeps a different probability vector over the set of experts and actions for each $x \in {\cal X}$, and updates these probability vectors only when the corresponding context is observed. The parameter $\kappa(x)$ counts the number of times context $x$ has occurred up to the current time.  Instead of $t$, $\kappa(x)$ is used to adjust the learning rate of each ranking expert that runs for different contexts. This way, each RE algorithm is guaranteed to achieve sublinear regret with respect to the best expert for its context.

The following theorem bounds the contextual regret of Algorithm~\ref{alg:CRE}.
\begin{theorem} 
\label{thm:CRE}
Assume that the UE uses the CRE algorithm with the pool of ranking experts ${\cal E}$ given in Definition \ref{def:CEXP4} with $\gamma_t =  \sqrt{ \frac{ \log ( N! +1) }{ t N} }$.  Then we have 
\begin{equation}
R_{\mathbf{a}}({\cal E}, T) \leq 2  N^2 \sqrt{T \log N} . 
\end{equation}
\end{theorem}

\begin{proof}
See Appendix~\ref{apd:thm5}.
\end{proof}

\section{Simulation Results}
\label{sec:sim}

\begin{table*}[h]
\caption{Simulation parameters}
\label{tab:sim_param}
\centering
\begin{tabular}{|l|l||l|l|} 
\hline
\textbf{Parameters} &  \textbf{Value}  & \textbf{Parameters} &  \textbf{Value}  \\ \hline
$N$ &  6, 12   & Pathloss model & 3GPP in-to-out \cite{3gpp.36.814} \\ \hline
$M$ &  6  &  Shadowing & log-normal with 5dB  variance  \\ \hline
Maximum UEs per SBS & 3  &  SBS transmit power & 15dBm \\ \hline
Energy threshold for 3GPP-macro & 10\%  &  FHO count threshold for 3GPP-FHO & 4 out of 20 \\ \hline
Thermal noise density  & -174dBm/Hz  &  UE noise figure  & 5.5dB \\ \hline
Carrier frequency  & 2.1GHz & Bandwidth & 20MHz \\ \hline
Penetration loss ($L_{ow}$)  & 10dB    &  $d_0$  & 1m \\ \hline
\end{tabular}
\end{table*}

In order to verify the proposed mobility management design, we resort to numerical simulations. In particular, a system-level simulator is developed in which the geometry of UE/SBS and the UE movement are explicitly modelled. Our simulator adopts the general urban deployment model in \cite{3gpp.36.814}. The simulation setting is created to highlight the FHO problem for stationary or slow-moving user, which is the focus of our paper. Specifically, we assume that there is a house of size $14 \times 14$ square meters in the middle of the simulated area, and $N$ SBSs are symmetrically placed around the room. Note that the symmetrical layout is made to speed up the simulations as well as to create a more severe FHO environment. The distance from each SBS to the center of the house is 80 meters. SBSs are transmitting at a fixed power of 15dBm. On average, there are $M$ UEs in the room, and we adopt a simple random waypoint mobility model \cite{Camp:02} with low speed to address user mobility. In particular, we trace the slow movement of one particular UE (the UE of interest), while allowing other UEs to randomly leave or enter the network, and move around with different serving SBSs. As a result, the UE of interest will see dynamic energy consumption from its varying serving SBSs.  The total energy consumption is normalized. We consider the 3GPP pathloss model that is recommended for system simulations of small cells and heterogeneous networks. Particularly, we consider the pathloss model suggested in \cite{3gpp.36.814}:
\begin{equation}
\label{eq:PLmodel}
PL(d) [dB] = 15.3 + 37.6 \times \log_{10}(d) + L_{ow}, d>d_0.
\end{equation}
Some other system simulation parameters are summarized in Table~\ref{tab:sim_param}.

We first study the energy consumption performance of the proposed BREW algorithm and compare with the existing 3GPP solutions described in Section~\ref{sec:3gpp}. In particular, we include both the original threshold-based handover rule and the enhanced FHO-aware policy, labeled as \textit{3GPP-macro} and \textit{3GPP-FHO}, respectively, in the plots. The SINR threshold for 3GPP-macro is set such that the corresponding normalized average energy consumption is above 10\%, with no additional offset. The enhanced FHO-aware solution adopts a freezing period that is of the same length as the BREW batch length for fair comparison. Furthermore, the threshold is set to be 4 handovers over the past 20 slots. 

\begin{figure}[htb]
    \centering
    \centerline{\includegraphics[width=0.8\textwidth]{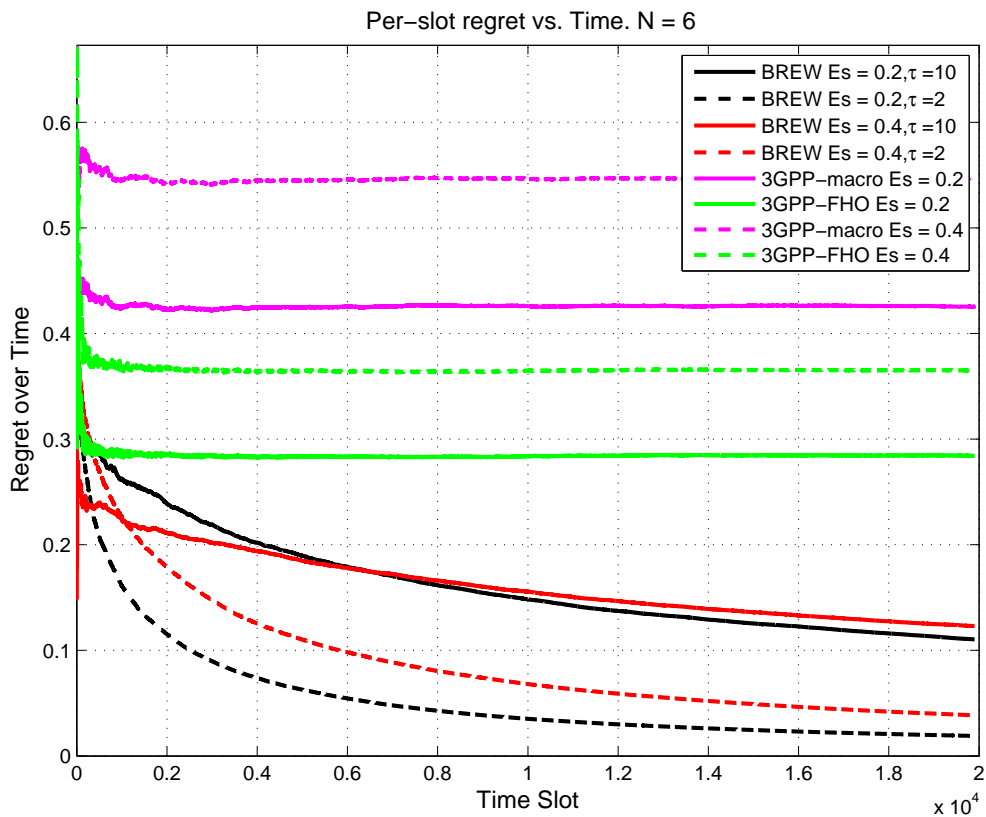}}
    \caption{Comparison of the per-time-slot energy consumption loss versus time for BREW, 3GPP-macro and 3GPP-FHO. $N=6$.}
    \label{fig:normal1}
\end{figure}

\begin{figure}[htb]
    \centering
    \centerline{\includegraphics[width=0.8\textwidth]{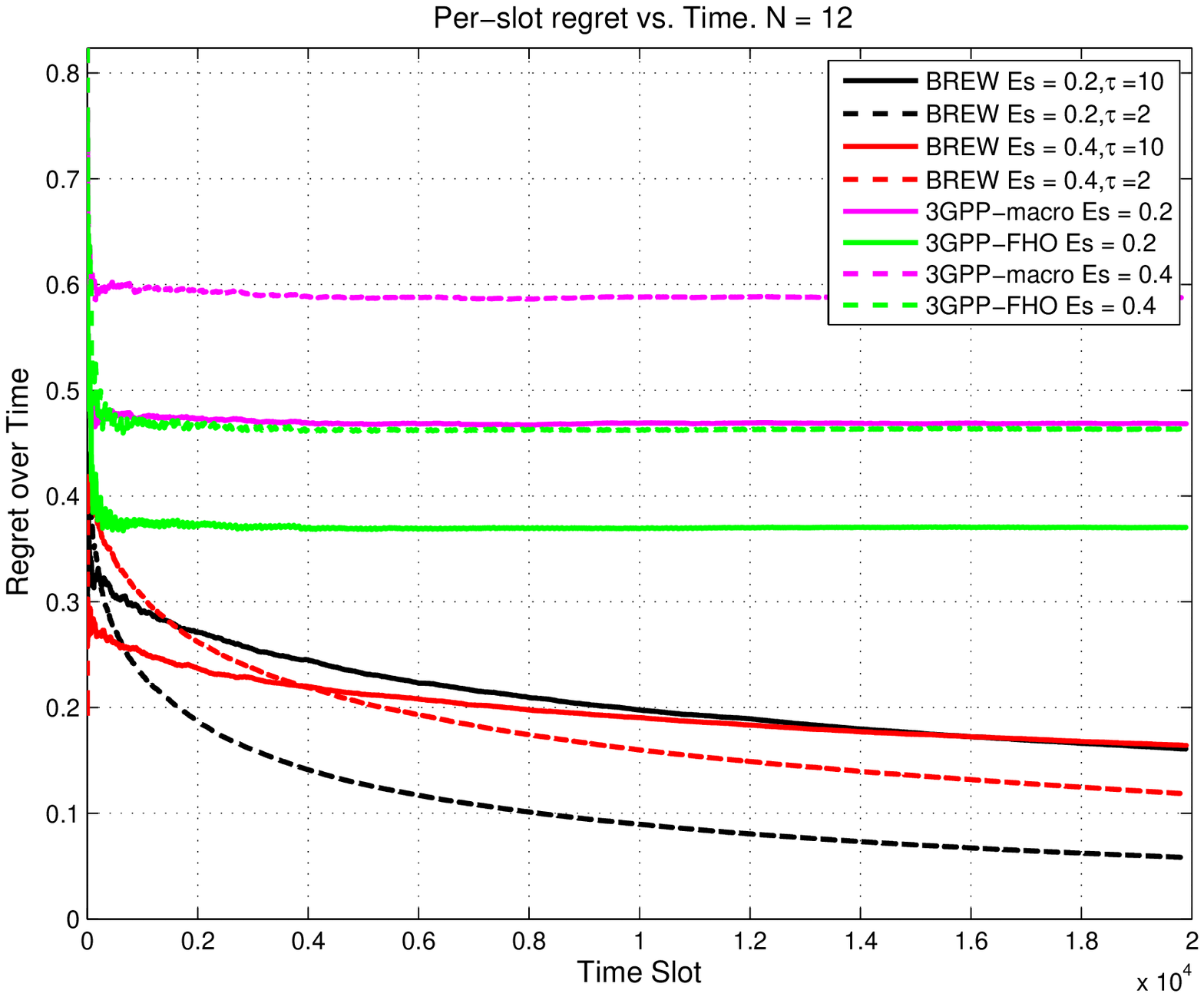}}
    \caption{Comparison of the per-time-slot energy consumption loss versus time for BREW, 3GPP-macro and 3GPP-FHO. $N=12$.}
    \label{fig:normal2}
\end{figure}

\begin{figure*}
\centering
\subfigure[$N = 6$, $E_s = 0.2$]{%
\includegraphics[width=0.45\textwidth]{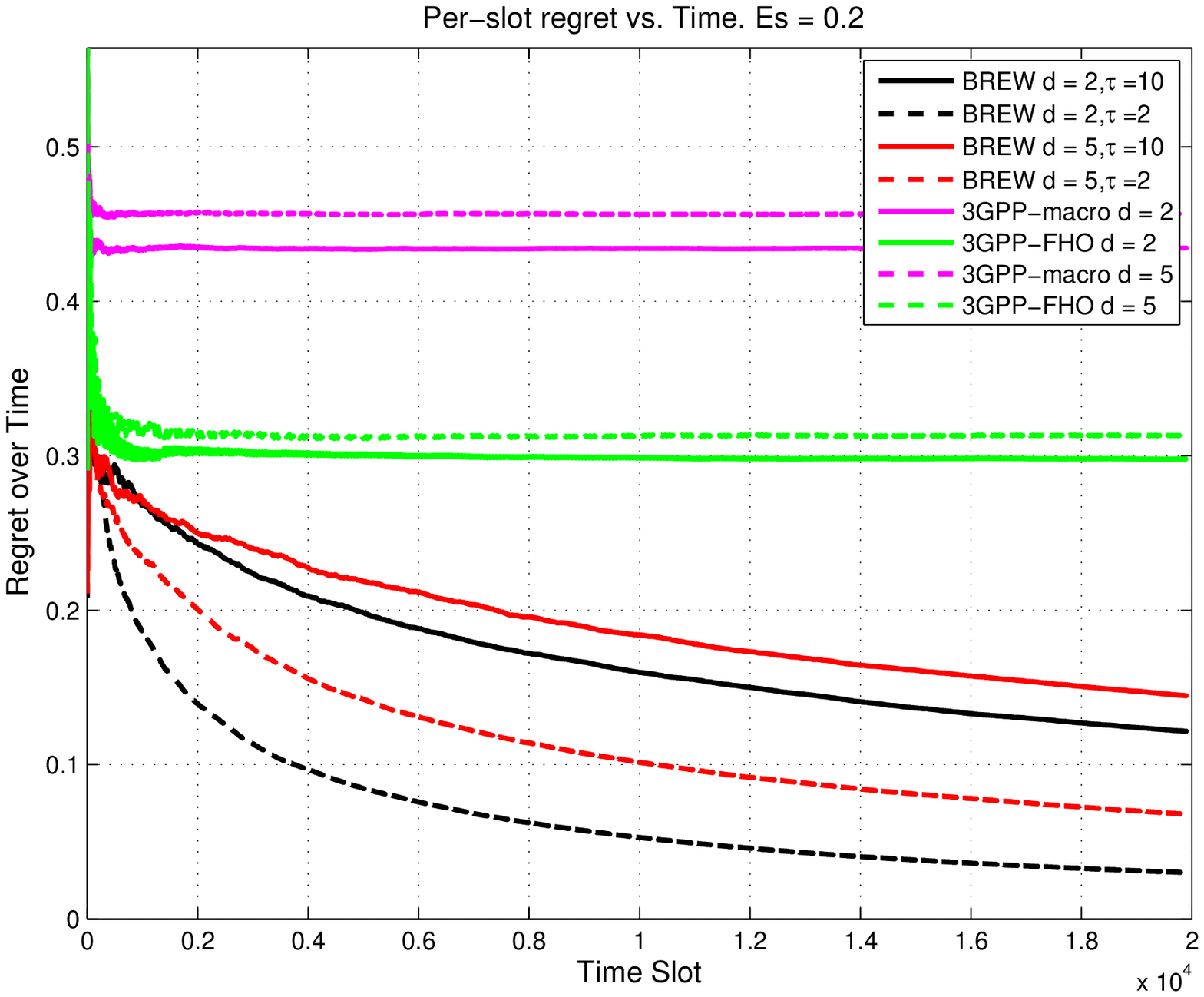}
\label{fig:simDelayN6Es2}}
\subfigure[$N = 6$, $E_s = 0.4$]{%
\includegraphics[width=0.45\textwidth]{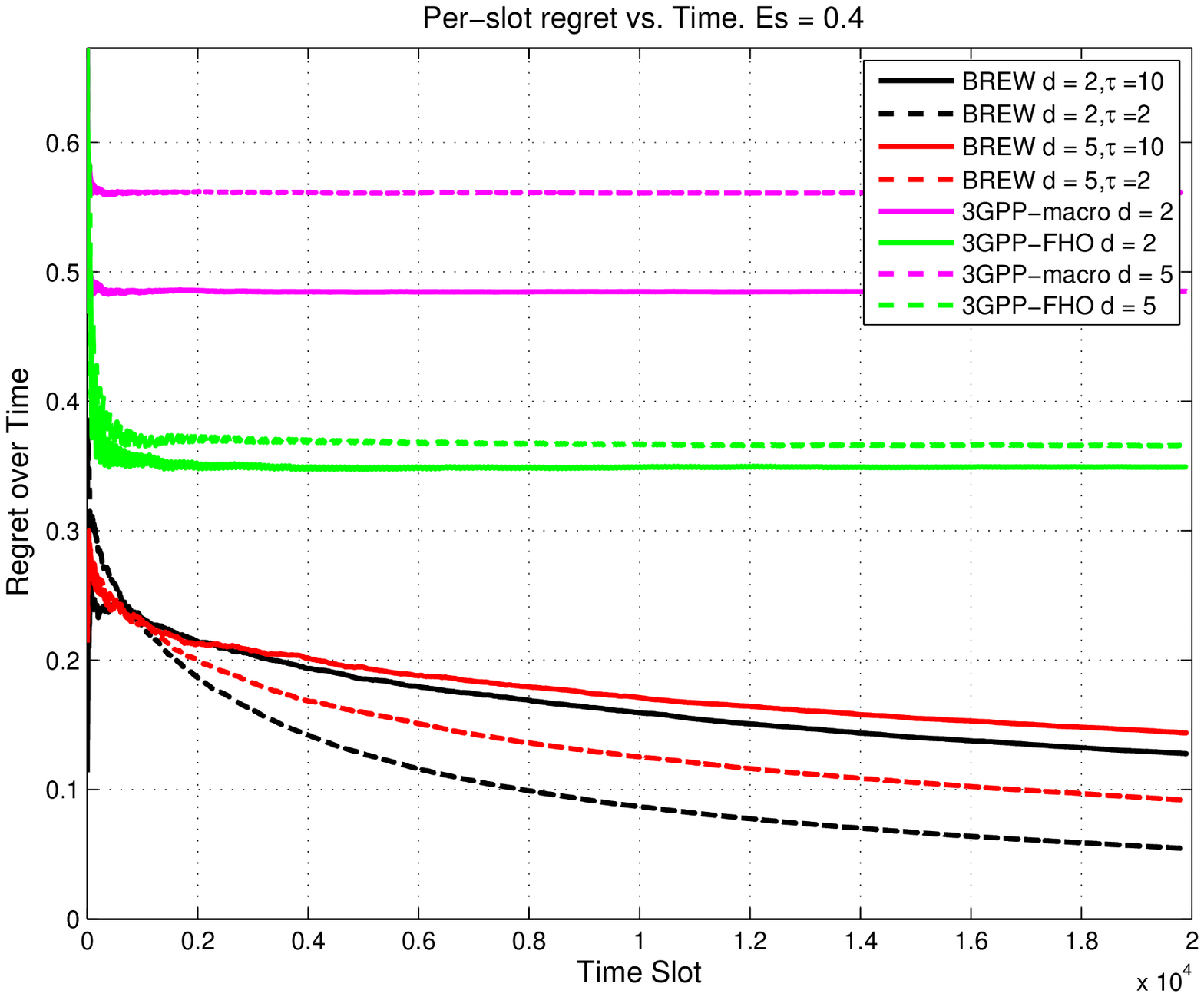}
\label{fig:simDelayN6Es4}}
\subfigure[$N = 12$, $E_s = 0.2$]{%
\includegraphics[width=0.45\textwidth]{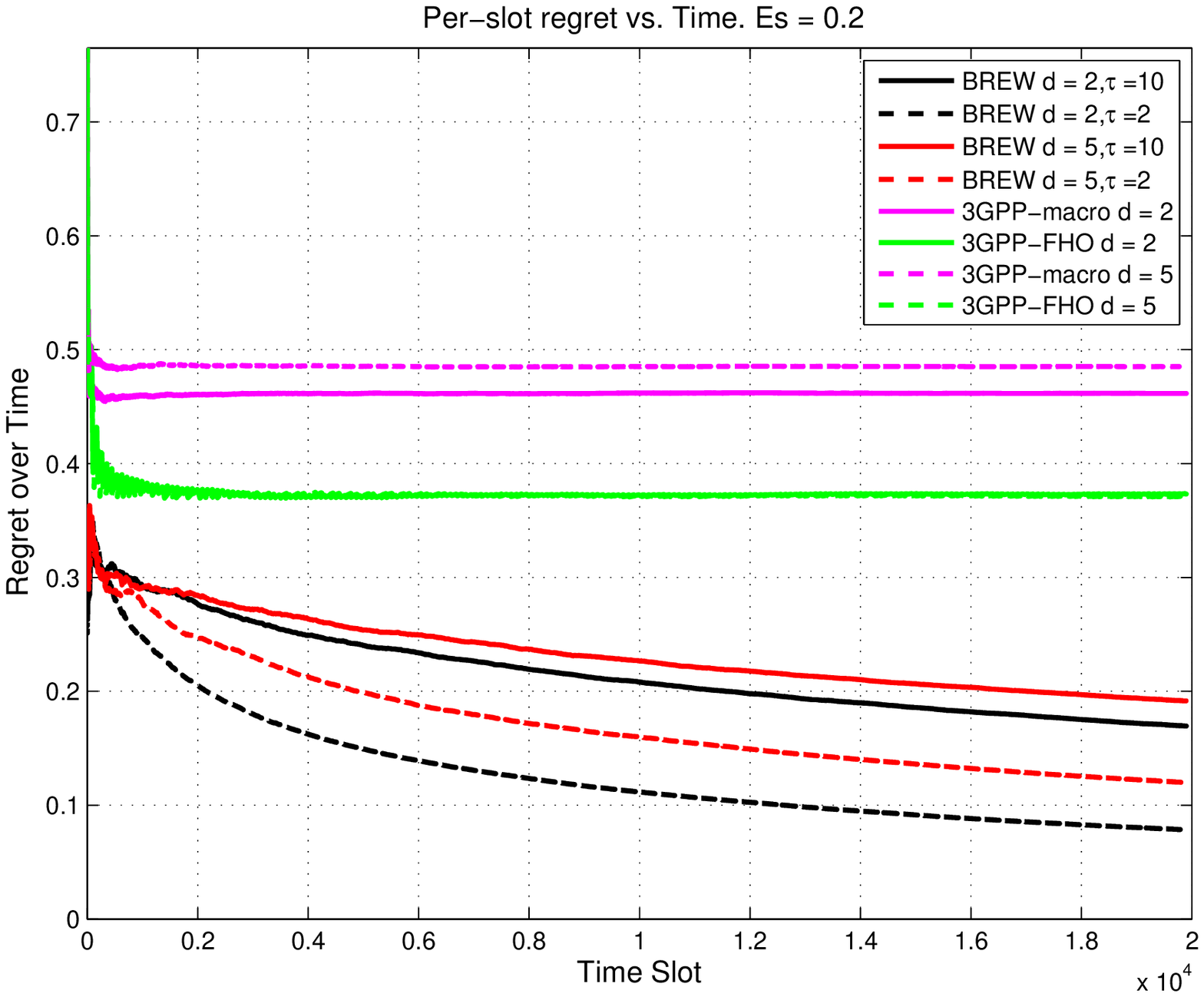} 
\label{fig:simDelayN12Es2}}
\subfigure[$N = 12$, $E_s = 0.4$]{%
\includegraphics[width=0.45\textwidth]{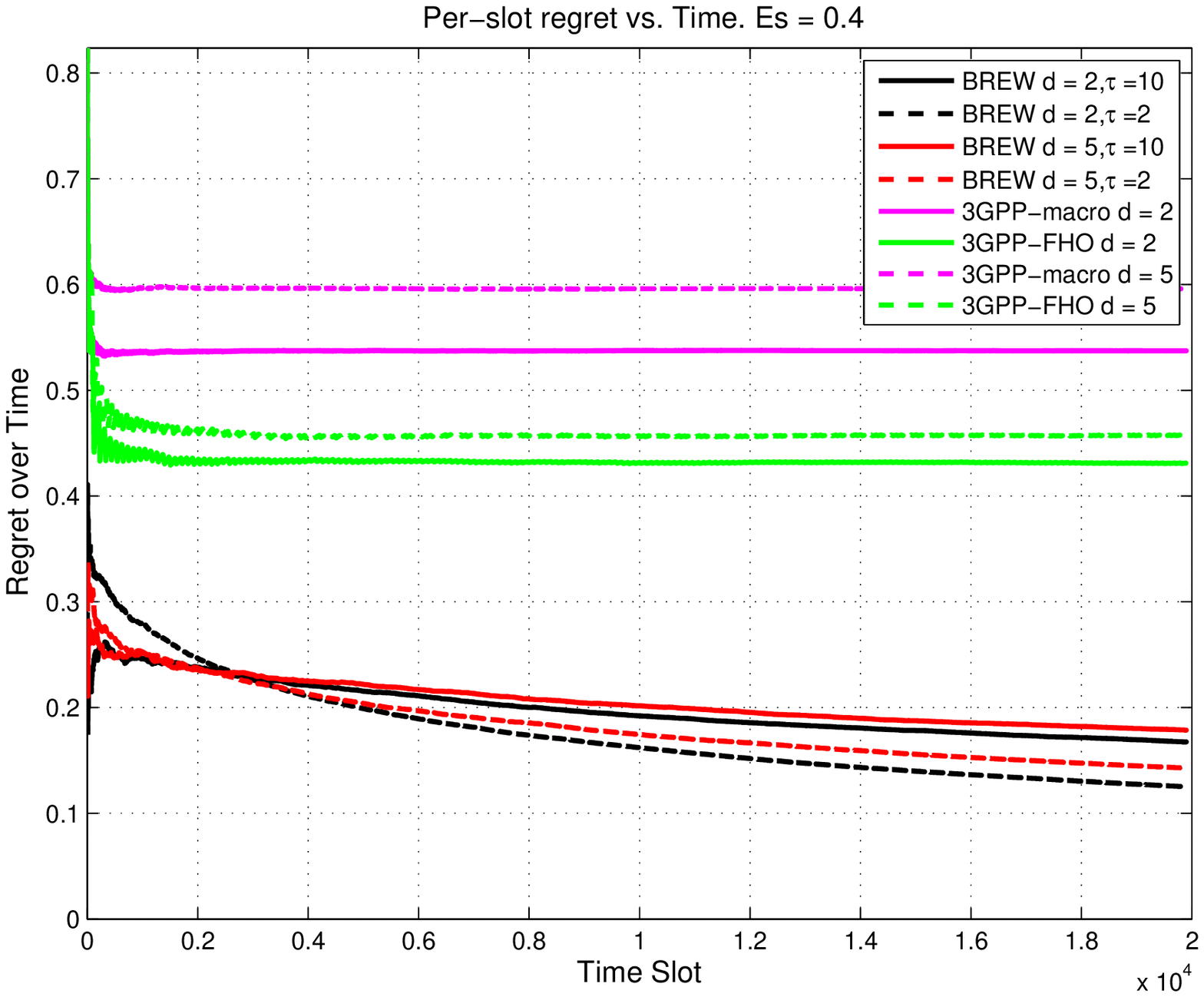}
\label{fig:simDelayN12Es2}}
\caption{Impact of delayed feedback to the regret for BREW, 3GPP-macro and 3GPP-FHO.}
\label{fig:simDelay}
\end{figure*}

The performance comparison is reported in Figure~\ref{fig:normal1} for $N=6$ and Figure~\ref{fig:normal2} for $N=12$, respectively, where the latter represents an extreme UDN deployment.  In particular, the total energy consumption of each algorithm is compared against the genie-aided solution where the UE selects the best SBS with the minimum energy consumption from the very beginning. The regret of each algorithm is normailized by time. A few important observations can be made from this system level simulation. First of all, we can see that 3GPP-FHO outperforms 3GPP-macro in terms of energy consumption, but both solutions do not exhibit a decaying per-slot regret, confirming the regret analysis in Section~\ref{sec:3gpp}. As a result, these solutions cannot converge to the optimal SBS asymptotically. The proposed BREW algorithm, however, has a diminishing per-slot regret and will converge to the best SBS, supporting the regret analysis in Section~\ref{sec:reg_BREW}. Second, the performance of existing solutions degrade significantly once the handover cost is explicitly taken into account, and such degradation remains constant over time. The increased handover cost also impacts the energy consumption of the proposed BREW algorithm, but thanks to the batched nature and the built-in exploration-exploitation tradeoff,  the amount of handovers will gradually reduce over time,  mitigating the effect of the increased handover cost. For the considered system simulations, the proposed BREW solution can achieve  $20\%-30\%$ less energy consumption (depending on the parameter setting) with a moderate time duration, and more than $60\%$ gain asymptotically, over the existing solutions. Finally, the effect of batch size $\tau$ can be analyzed from the figures, which reveals the inherent tradeoff between quick exploration (and hence finding the optimal SBS faster) and the handover cost associated with such exploration. As we can see, for both $N=6$ and $N=12$, there exists an initial period where a large batch size results in less energy consumption. This is because in the initial time slots, a small batch size would lead to more frequent handovers for exploration, which results in both more handover costs and selecting sub-optimal SBSs more. However, as time goes by, the speed of exploration slows down, and we will enter a separate region where a large batch size leads to more time spent on sub-optimal SBSs, which increases the energy consumption.

\begin{figure*}[h]
\centering
\subfigure[$N = 6$, $E_s = 0.2$]{%
\includegraphics[width=0.45\textwidth]{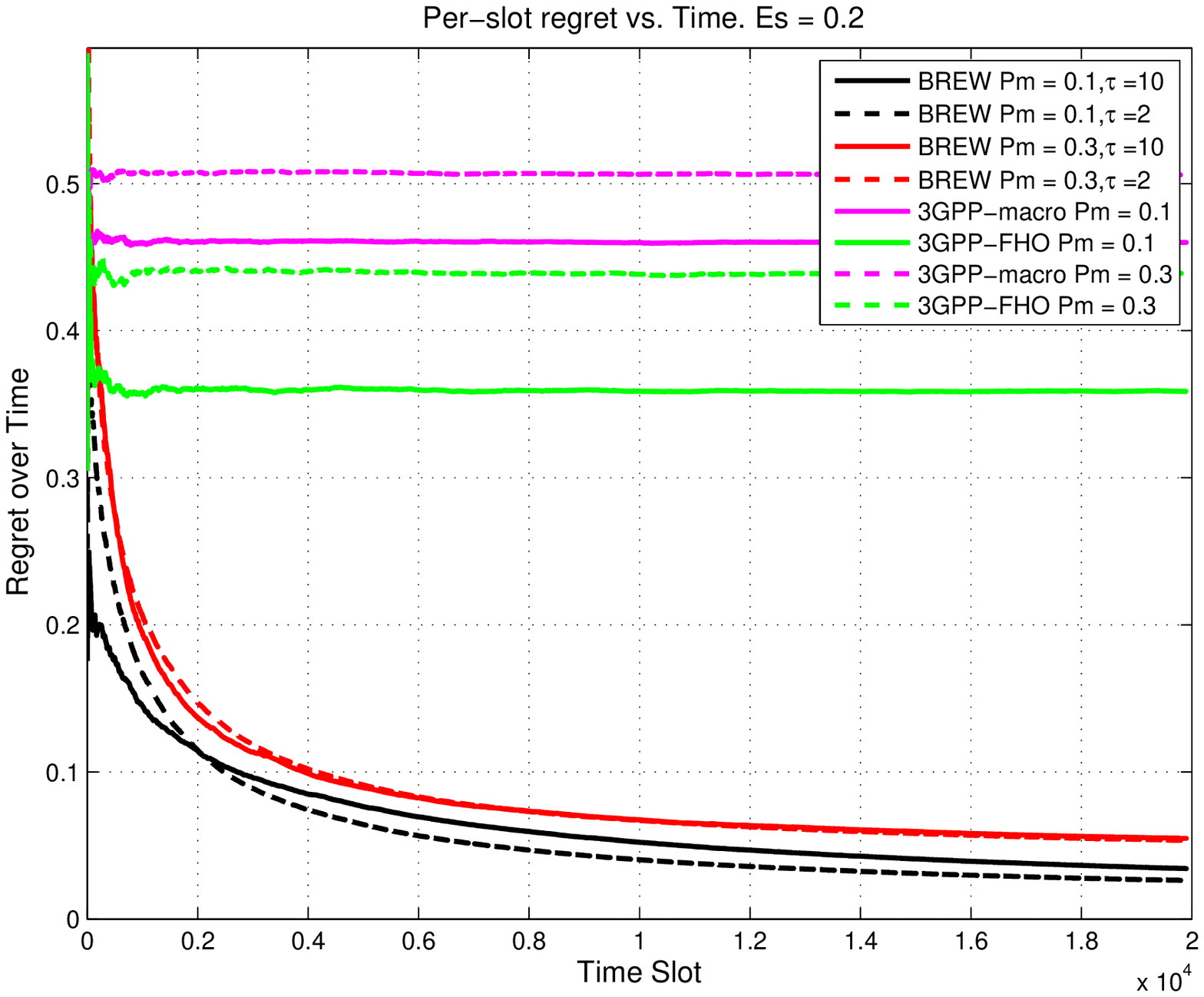}
\label{fig:simMissN6Es2}}
\subfigure[$N = 6$, $E_s = 0.4$]{%
\includegraphics[width=0.45\textwidth]{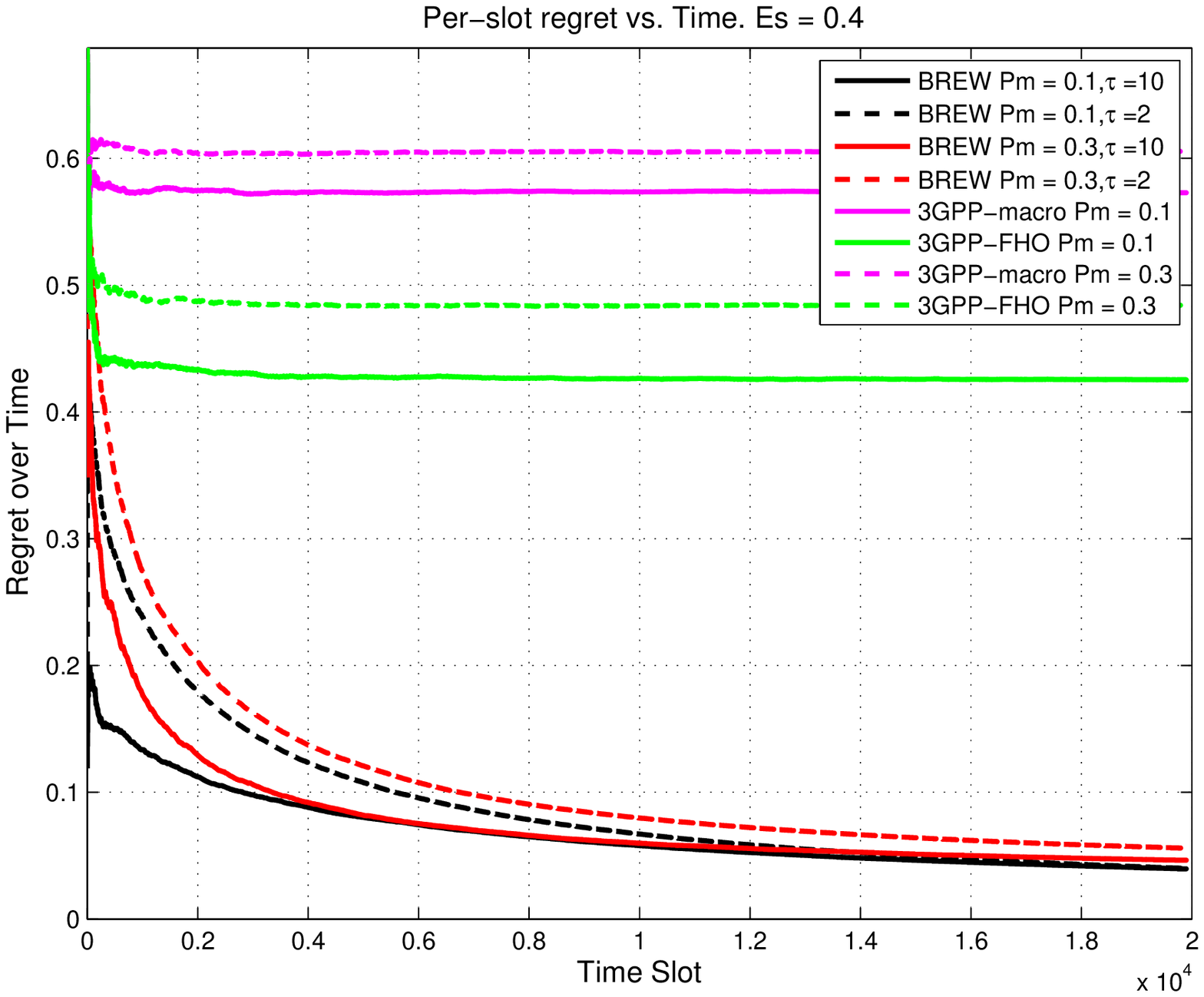}
\label{fig:simMissN6Es4}}
\subfigure[$N = 12$, $E_s = 0.2$]{%
\includegraphics[width=0.45\textwidth]{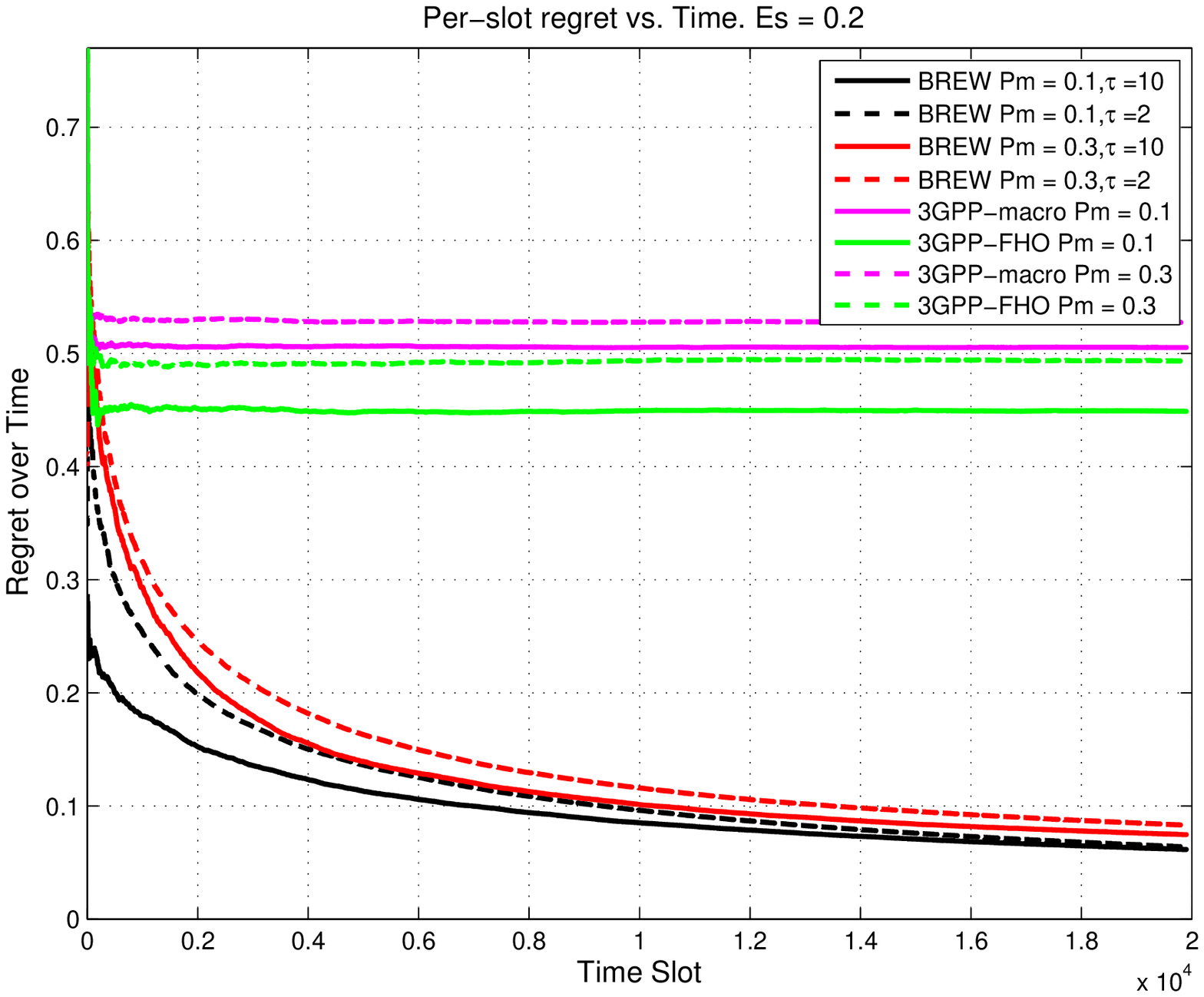} 
\label{fig:simMissN12Es2}}
\subfigure[$N = 12$, $E_s = 0.4$]{%
\includegraphics[width=0.45\textwidth]{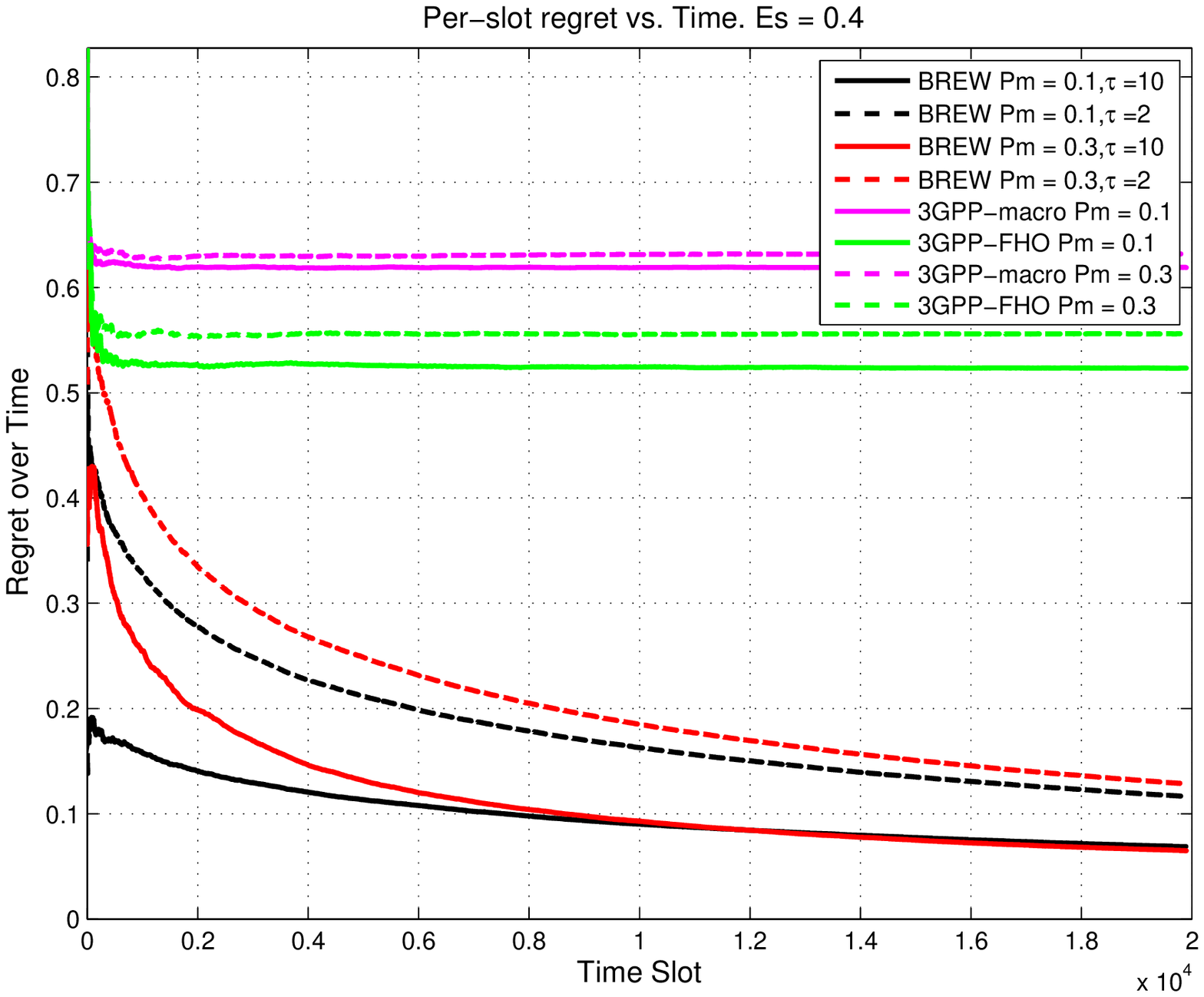}
\label{fig:simMissN12Es2}}
\caption{Impact of missing feedback to the regret for BREW, 3GPP-macro and 3GPP-FHO.}
\label{fig:simMiss}
\end{figure*}

Figure~\ref{fig:simDelay} studies the impact of delayed feedback on the performance of the three algorithms. We can see that additional delays of sending the energy consumption feedback increases the regret for all algorithms, and larger delay leads to more severe regret increase. However, an important observation from Figure~\ref{fig:simDelay} is that the impact of delayed feedback on BREW is mild when the batch size is moderate or the handover cost is large. This is due to the fact that when the batch size is not very small, the handover decision will be slightly postponed due to the delayed feedback, and the UE of interest stays on the same SBS while waiting for feedback. Because the accumulated feedback comes from a batch, a slight offset will not significantly alter the averaged feedback of the batch, which provides robustness against information obsolete. Additionally, the larger handover cost will further penalize myopic protocols, where the handover decisions are based on outdated information. The simulation results show the robustness of the BREW algorithm against delayed feedback.

Similarly, Figure~\ref{fig:simMiss} reports the simulation results when each time slot, the feedback energy consumption may be missing following an i.i.d. Bernoulli model with missing probability $P_m$. The BREW algorithm used in the simulation is the extended version as in Algorithm~\ref{alg:BREWmiss} that considers $P_m$. It can be concluded that missing feedback impact all three algorithms in terms of the regret performance, but with different behavior. For the two 3GPP solutions, the regret quickly converges and there exists an almost constant gap asymptotically. For the extended BREW algorithm, however, there exists a rather large gap during the initial period. This is due to the lack of accurate information and hence  missing feedback have a bigger impact to the regret. As the algorithm gradually learns the loss information, the impact of missing feedback diminishes, as shown by the very small gap between different $P_m$ values for large $t$. 

\begin{figure*}
    \centerline{
        \subfigure[ $E_s=0.2$ ]{ 
        \includegraphics[width=0.45\textwidth]{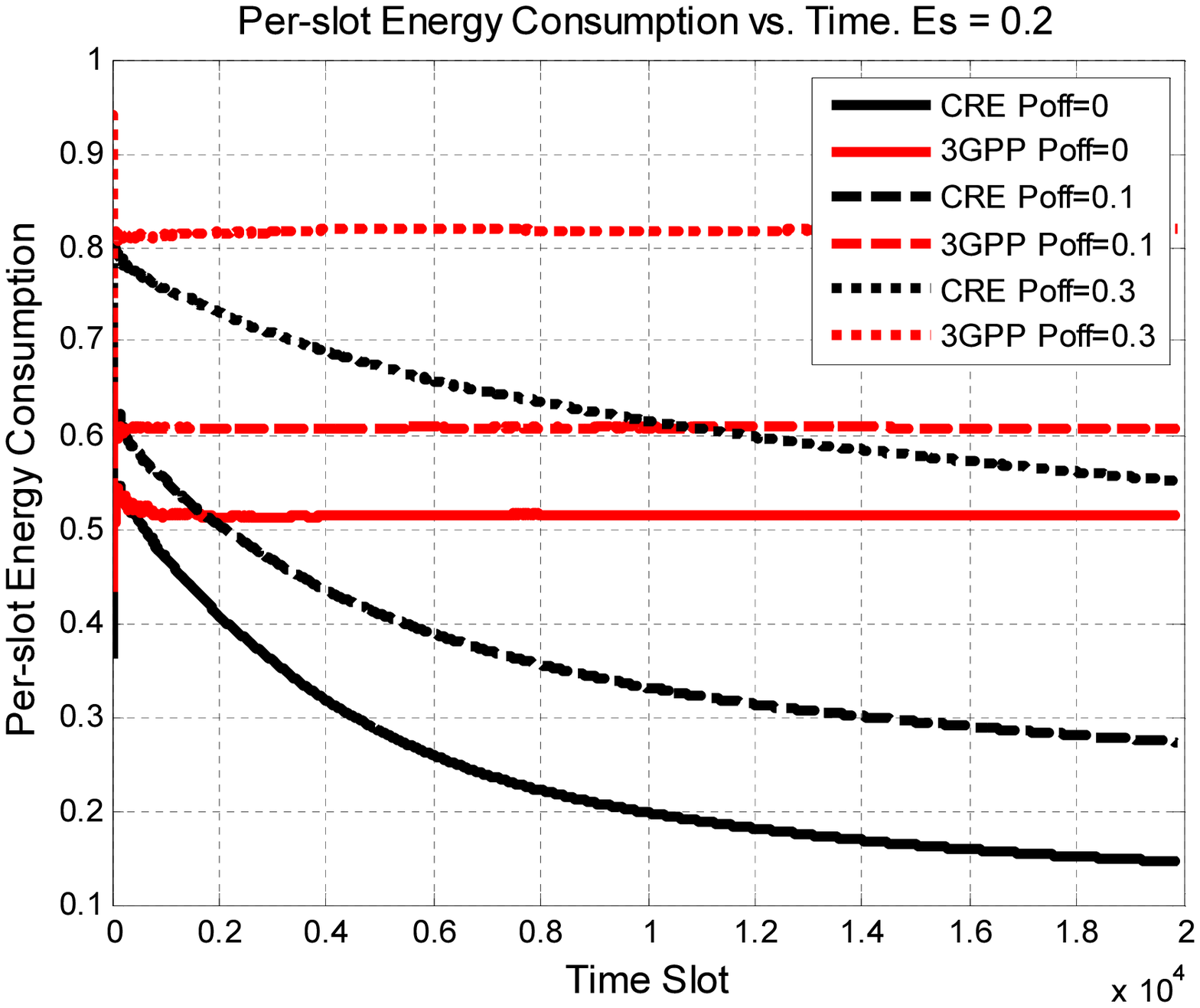}
        \label{fig:onoff1}}
        \hfil
        \subfigure[ $E_s=0.4$ ]{         
        \includegraphics[width=0.45\textwidth]{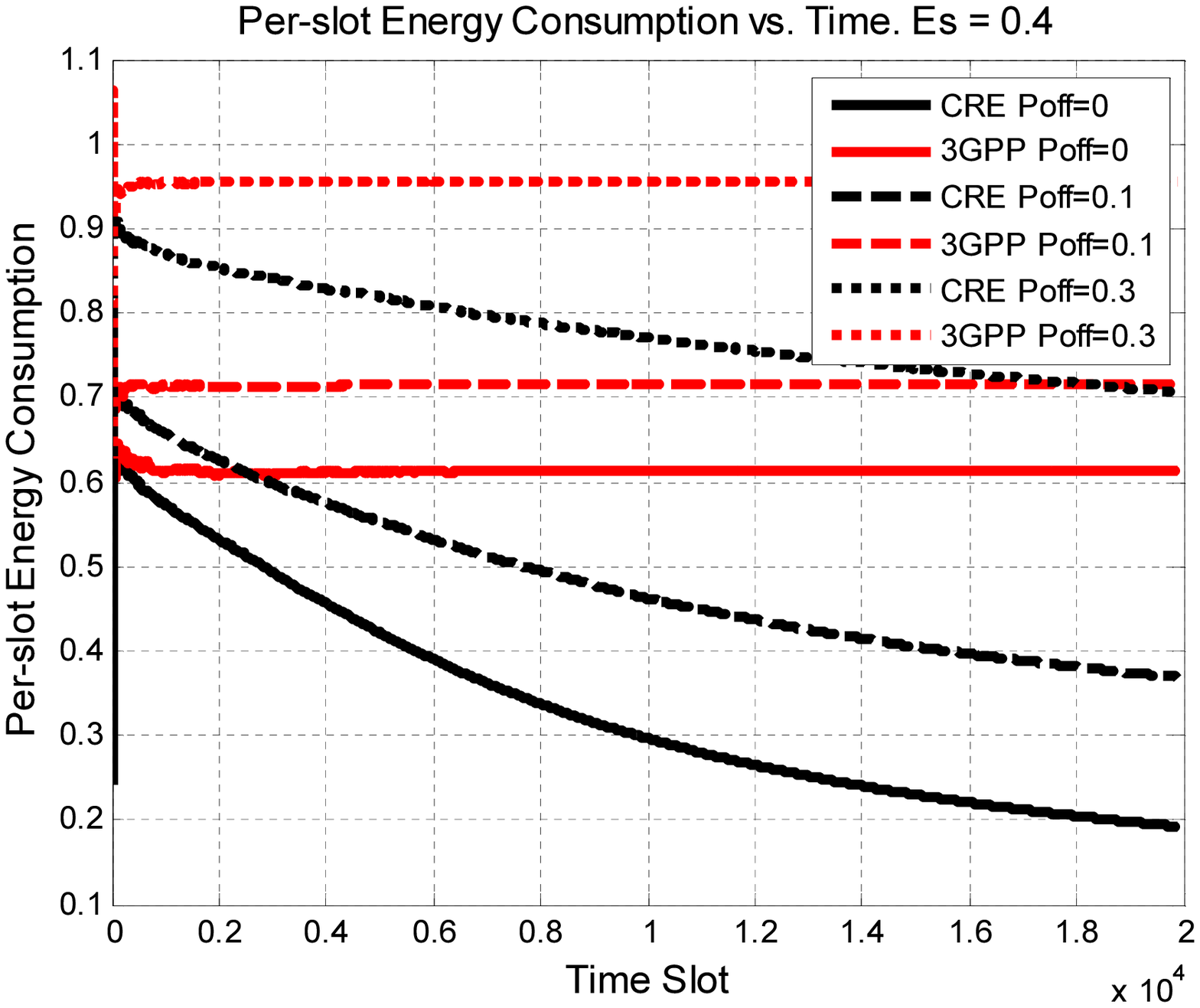}
        \label{fig:onoff2}}}
    \caption{Impact of dynamic SBS on/off to the per-time-slot energy consumption versus time for CRE and 3GPP-macro.}
    \label{fig:onoff}
\end{figure*}

Finally, we study the impact of dynamic SBS on/off to the mobility algorithms. We assume that the $N$ neighboring SBSs are installed by end-users and  they can be turned on and off at the users' discretion. To model the dynamic SBS presence, we assume that at each time slot, all $N$ SBSs may independently be turned on or off following an identical Bernoulli distribution with off-probability $P_{\sf{off}}$. It is worth noting that this model presents a much bigger challenge than models where SBS dynamics have patterns. As has been discussed, the CRE algorithm presented in Section~\ref{sec:onoff} achieves a good tradeoff between complexity (number of experts) and performance. We compare the average per-slot energy consumption of CRE to 3GPP, which at each slot selects the best SBS that is not turned off. For a fair comparison, we also assume that the 3GPP metric takes into account the potential handover cost $E_s$, which is not considered in the standard 3GPP-macro algorithm. Figure~\ref{fig:onoff} presents the numerical comparison of these two algorithms with different $P_{\sf{off}}$ values. Clearly, dynamic SBS affects both algorithms, but CRE quickly outperforms the extended 3GPP-macro algorithm and the per-slot energy consumption decreases as time goes by. This is due to the gradual convergence to the best expert in the expert pool.


\section{Conclusions}
\label{sec:conc}

Emerging wireless networks have become more heterogeneous and the network density has increased significantly, both of which pose significant challenges to energy efficient mobility management. Existing solutions, mostly based on optimizing immediate system objectives, fail to achieve long-term minimum energy consumption in highly dynamic and complex wireless networks. To address this problem, we have made two novel contributions. The first is that we adopt a \textit{non-stochastic} online-learning approach to model the UDN mobility management. The key benefit of this approach, as its name suggests, is that we do not need any assumption on the statistical behavior of the SBS activities. This is extremely desirable for UDN. The other novelty is that we explicitly add the \textit{handover cost} to the utility function, which forces the resulting solution to minimize frequent handovers. 

Built upon these two key ideas, we have proposed the BREW algorithm which relies on batching to explore in bulk, thus reducing the handovers that are typically required for exploration. A sublinear regret upper bound for BREW is proved. We then study how the BREW algorithm can be adjusted to deal with various system imperfections, including delayed or missing feedback. Most importantly, we have studied the  impact of dynamic SBS on/off, which often arises in user-deployed small cell networks. We first prove an impossibility result with respect to any arbitrary SBS on/off. Then, a novel strategy, called ranking expert (RE), is proposed to simultaneously address the handover cost and the availability of SBS. The complete RE algorithm results in a large number of experts, which incurs significant complexity. We further propose a contextual ranking expert (CRE) algorithm that reduces the number of experts significantly. Regret bound is proved for both RE and CRE with respect to the best expert. Simulation results show a significant improvement to the overall system energy consumption. More importantly, the gain is robust against various system dynamics.

There are some interesting problems that have not been fully addressed in this work, which are the subjects of potential future work. For example, the regret upper bounds developed in this work are mostly for a given set of algorithm parameters, for which sublinear regret is rigorously proven. It is of interest to study the performance bound variation and tightness with respect to the algorithm parameters. Another important question is how to further enhance the ranking expert solutions in Algorithm~\ref{alg:RE} and \ref{alg:CRE}, in terms of the algorithm complexity and the corresponding regret bound.

\appendices
\section{Proof of Theorem~\ref{thm:SwitchRegret} }
\label{apd:thm1}

We cite two known results in non-stochastic bandit theory that will be used in the proof. These results are modified to fit into the problem setting of Theorem~\ref{thm:SwitchRegret}.

\begin{prop}
\label{prop1} (Part of Theorem 2.1 in \cite{bubeck2010jeux})
The standard pseudo-regret bound for the any-time EXP3 algorithm \cite{bubeck2010jeux}, where the parameter $\gamma_t$ does not depend on the time horizon $T$, is given by
\begin{equation}
\label{eqn:apd1}
{R}(T) \leq \sqrt{4.5 T N \log N},
\end{equation}
for $\gamma_t = \sqrt{ (2 \log N)/ (t N) }$, $t \in \mathbb{N}_{+}$.
\end{prop}

\begin{prop}
\label{prop2} (Part of Theorem 2 in \cite{arora2012online})
When $\tau > 1$, the regret $R'(T)$ of an algorithm with respect to the constant actions when the reward sequence is generated by an $m$-memory-bounded adaptive adversary is bounded by 
\begin{equation}
\label{eqn:apd2}
R'(T) \leq \tau R(\frac{T}{\tau}) + \frac{Tm}{\tau} + \tau
\end{equation}
where $R(T)$ denotes an upper bound on the standard pseudo-regret of the algorithm.\footnote{Definitions of $m$-memory-bounded adaptive adversary and standard pseudo-regret can be found in \cite{arora2012online}.}
\end{prop}

The proof of Theorem~\ref{thm:SwitchRegret} then follows by recognizing that the adversarial bandit problem with switching costs is a special case of a $m=1$ memory bounded adversary. With $\tau = \lceil B_N T^{1/3} \rceil$, we have
\begin{equation}
{R}\left(\frac{T}{\tau}\right) \leq  \sqrt{  \frac{T B_N^{-3}}{\lceil B_N T^{1/3} \rceil} } 
 \leq T^{1/3} B_N^{-2},  \label{eqn:apd3}
\end{equation}
and thus
\begin{eqnarray}
R_{\mathbf{a}}(T)  &\leq&   \tau   T^{1/3} B_N^{-2} + \frac{T}{\tau} + \tau \nonumber \\
 &\leq& 2 B^{-1}_N  T^{2/3} + \left( B_N + B^{-2}_N  \right) T^{1/3} + 1.
\end{eqnarray}

\section{Proof of Theorem~\ref{thm:DelayRegret} }
\label{apd:thm2}

In \cite{arora2012online} a $d$-memory-bounded adversary is defined as an adversary which is restricted to choose a loss function that depends on the $d+1$ most recent actions of the learner. Hence, the loss of the learner generated by a $d$-memory-bounded adversary at time slot $t$ can be written as $f_t( a_{t-d}, \ldots , a_t )$. Now consider our setting in which the energy cost of choosing an SBS at time $t$ only depends on the state of the network at time $t$. This cost is given by the function $\{  E_t(a) \}_{ a \in {\cal N}_{\text{SBS} } }$. When the feedback is received by the learner with a delay of $d$ time slots, we can model the cost incurred by the learner as cost assigned by a $d$-memory-bounded adversary whose loss function is 
\begin{equation}
f_t( a_{t-d}, \ldots , a_t ) = E_{t-d}(a_{t-d}) + E_s \mathds{1}_{\{a_{t-1} \neq a_t\}}
\end{equation}
The result then follows the same steps as Appendix~\ref{apd:thm1}.

\section{Proof of Proposition~\ref{prop:3gpp1}}
\label{apd:3gpp1}

We separately prove the linear regret in $T$ for stochastic and non-stochastic energy consumption models.  For a stochastic model, we denote the average RSRP or RSRQ of each SBS as $r_{n} = \bb{E} [R_t(n)]$, and the probability that the metric of SBS $n$ falls below the pre-determined threshold $\theta$ as $\sigma_{n} = P(R_t(n) < \theta)$. Without loss of generality and to avoid trivial conditions, we assume that $r_{1} >r_{2} > \cdots >r_{N}$. 

A regret lower bound for the 3GPP handover protocols described in Section~\ref{sec:3gpp} can be achieved by a genie-aided policy where switching only happens between SBS $1$ and $2$. In other words, whenever the performance metric of the best SBS falls below $\theta$, the user only switches to the second-best SBS; when the performance metric of the second-best SBS falls below $\theta$, it comes back to the best SBS. This policy can be modelled as a two-state Markov process where each state $n$ has a transition probability $\sigma_{n}$. We denote the steady-state distribution for the sub-optimal SBS $2$ as $\rho_2$, and it can be shown $\rho_2 > 0$ for non-trivial cases. Thus, the genie-aided policy achieves a linear regret $\rho_2 (r_{1}-r_{2})T $ asymptotically.

For a non-stochastic model, we prove the linear regret in $T$ by constructing a specific sequence of  metrics $\{R_{t}(n)\}$. For simplicity, we only give one example for $N=2$. Consider $R_{1}(1) > R_{1}(2) $ so that at time slot $1$ the best SBS is selected. Then we let $R_{2}(1) < \theta < R_{2}(2)$ so that the UE switches to the sub-optimal SBS. We then fix $R_{t}(1) >R_{t}(2) > \theta$ for all $t>2$. In this example, the UE will be stuck with the sub-optimal SBS $2$ from time slot 2 to $T$, thus achieving a linear regret in $T$.

\section{Proof of Theorem~\ref{thm:linearregret} }
\label{apd:thm3}

The adversary defined in Theorem~\ref{thm:linearregret} generates ${\cal N}_{t+1}$ based on $a_t$. Consider the worst-case scenario where it simply lets ${\cal N}_{t+1} = {\cal N}_{ \text{SBS} } - \{ a_t \}$. 
This forces the UE to switch at every time slot. Hence it incurs a handover loss of $E_s T$. As a result, the loss of the learning algorithm is at least $E_s T$.

Since only one SBS is inactive in each time slot, there exists at least one SBS which is active in at least $T (1 - 1/N)$ time slots. Let $\tilde{a}$ denote such an SBS. SBS $\tilde{a}$ will be inactive in at most $T/N$ time slots, which means that any policy that selects SBS $\tilde{a}$ when it is available needs to switch at most $T/N$ times. Thus, the cost of such a policy is bounded above by $T (1 - 1/N) E_{\max} + T / N $, where  the first term denotes the worst-case energy consumption from $\tilde{a}$ at time slots when it is active and the second term denotes the worst-case energy consumption plus handover cost due to the slots in which $\tilde{a}$ is inactive\footnote{Recall that the normalized energy consumption in a time slot is upper bounded by $1$.}. As a result, we have that the cost of $\tilde{a}$ is bounded above by $T (1 - 1/N) E_{\max} + T / N $.

Let $\tilde{a}^*$ denote the best SBS (the one whose cumulative loss is minimum). Then, the loss of $\tilde{a}^*$ is upper bounded by $T (1 - 1/N) E_{\max} +  T / N$. 

Hence, the difference between the loss of the learning algorithm and the loss of $\tilde{a}^*$ is at least 
\begin{eqnarray}
&& E_{{s}} T -  T (1 - \frac{1}{N}) E_{\max}-   \frac{T}{N}  \\ & = & \frac{T}{N} (N E_s - (N -1) E_{\max} - 1 )    \\
& \geq & \frac{T E_{\max}}{N} 
\end{eqnarray}
where the inequality follows from $E_s \geq E_{\max} + 1/ (N-1)$. This proves that the regret is linear in $T$.

\section{Proof of Theorem~\ref{thm:ranking} }
\label{apd:thm4}

Note that the SBS activity evolves independently of the actions of the UE. Hence, the adversary is only able to modify the current reward of the UE based on its current action. Hence, the adversary is oblivious to the actions of the UE. 
Therefore, we can use Theorem 4.2 in \cite{Bubeck:12} to bound the regret. The number of experts including the uniform expert is $(N!)^N + 1$. We obtain the result by observing that
\begin{eqnarray}
\log (  (N!)^N + 1 ) \leq  \log (  (N! + 1 )^N ) &=& N \log (N!+ 1) \nonumber \\ &\leq& N^2 \log N,
\end{eqnarray}
where the last inequality comes from $\log(N! +1 ) \leq N \log N$.

\section{Proof of Theorem~\ref{thm:CRE}}
\label{apd:thm5}


We have that $\{  \tau_b \}_{b \in {\cal N}_\text{SBS}}$ is a random variable which depends on the randomization and the history of actions selected by CRE. Our regret bound will hold for any realization of $\{  \tau_b \}_{b \in {\cal N}_\text{SBS}}$. 
Consider the loss function 
\begin{equation*}
E_t( a_t ) + E_s \mathds{1}_{\{a_t \neq b \}}. 
\end{equation*}
By the definition of contextual regret and because $E_t$ is generated by an oblivious adversary, for any $b \in {\cal N}_\text{SBS}$ we have 
\begin{eqnarray}
 \sum_{t \in \tau_b }  \mathbb{E} \left[  E_t( a_t ) + E_s \mathds{1}_{\{a_t \neq b \}} \right] -  
  \min_{e \in {\cal E} }  \sum_{t \in \tau_b } \left[ E_t( a^e_t(b) ) + E_s \mathds{1}_{\{a^e_t(b) \neq b\}}  \right]   &\leq & 2 \sqrt{ |\tau_b | N \log (N! + 1 ) } \label{eqn:CEXP4bound1}    \\
 & \leq & 2 \sqrt{ |\tau_b | N^2 \log (N ) }    \\
 & \leq & 2 \sqrt{ T N^2 \log (N ) } \label{eqn:CEXP4bound}
\end{eqnarray}
where (\ref{eqn:CEXP4bound1}) comes from the standard regret bound of EXP4, and the expectation is taken with respect to the randomization of CRE when the context is $b$. Although 
$\{  \tau_b \}_{b \in {\cal N}_\text{SBS}}$ depends on the randomization of CRE, since the bound derived in (\ref{eqn:CEXP4bound}) is independent of the randomization of CRE, we get the final result by summing (\ref{eqn:CEXP4bound}) over all $b \in {\cal N}_\text{SBS}$.

\bibliographystyle{IEEEtran}
\bibliography{mobility}

\end{document}